\documentclass[format=sigconf,authorversion,nonacm]{acmart} %

\usepackage{multicol}
\usepackage{esvect}
\usepackage{float}
\usepackage{environ}
\usepackage{pdfpages}
\usepackage{centernot}
\usepackage{scalerel}
\usepackage[export]{adjustbox}
\usepackage{xifthen}
\usepackage{refcount}
\usepackage{listings}
\usepackage{amsthm}
\usepackage{xspace}
\usepackage{latexsym}
\usepackage{stmaryrd}
\usepackage{mathtools}
\usepackage[inference]{semantic}
\usepackage{etoolbox}
\usepackage{tikz}
\usepackage{xstring}
\usepackage{alphalph}
\usepackage[noend]{algpseudocode}
\usepackage{hyperref}
\usepackage{bm}
\usepackage{soul}
\usepackage{caption}
\usepackage{subcaption}
\usepackage{flexisym}
\usepackage{breqn}
\usepackage{mathpartir}
\usepackage{thmtools}

\usepackage[capitalise, noabbrev]{cleveref}

\makeatletter
\newcommand*\NoIndentAfterEnv[1]{%
  \AfterEndEnvironment{#1}{\par\@afterindentfalse\@afterheading}}
\makeatother

\definecolor{figuresmain}{RGB}{200,200,200}
\definecolor{figuretext}{RGB}{0,0,0}

\NoIndentAfterEnv{lstlisting}

\newcommand{\llbar}{\{\kern-0.5ex|}
\newcommand{\rrbar}{|\kern-0.5ex\}}

\newtheoremstyle{customstyle}%
  {3pt}%
  {15pt}%
  {\normalfont}%
  {0pt}%
  {\bf}%
  {.}%
  { }%
  {}%

\theoremstyle{customstyle}
\ifcsname definition\endcsname \else \newtheorem{definition}{Definition} \fi
\ifcsname theorem\endcsname \else  \fi
\ifcsname lemma\endcsname \else \newtheorem{lemma}{Lemma} \fi
\ifcsname corollary\endcsname \else \newtheorem{corollary}{Corollary} \fi
\ifcsname example\endcsname \else \newtheorem{example}{Example}[] \fi

\newtheorem*{theorem*}{Theorem}
\newtheorem*{lemma*}{Lemma}

\newcommand*{\ttfamilywithbold}{\fontencoding{T1}\fontfamily{pcr}\selectfont}

\lstdefinelanguage{mungo}{
  	comment=[l]{//},
    keywords = [1]{enum, class, switch, if, new, true, false, continue, loop, unit, return, infer, else, string, int, none, fun, val, this, rec, end},
    keywords = [2]{},
    keywords = [3]{double, int, bool, void},
    morekeywords=[3]{BankAccount, DataStorage, SalaryManager}
}

\lstdefinelanguage{scalaprotocol}{
  	comment=[l]{//},
    keywords = [1]{object, extends, with, class, var, def, val},
    keywords = [2]{goto, when, in, end, @Typestate},
    keywords = [3]{App, String, Double, Int, Float, Unit, Float},
    morekeywords={[3]{BankAccountProtocol, ProtocolLang, BankAccount, Instance, Alias, SalaryManager, DataStorage, Demonstration}},
    morestring=[b]"
}

\lstdefinelanguage{plaid}{
  	morecomment=[l][\itshape]{//},
    keywords = {val, state, new, with, method, this, case, of}
}

\definecolor{cyan}{HTML}{00afaf}
\definecolor{blue}{HTML}{0087ff}
\definecolor{green}{HTML}{859900}
\definecolor{base01}{HTML}{586e75}
\definecolor{base00}{HTML}{657b83}
\definecolor{orange}{HTML}{cb4b16}

\lstdefinestyle{color}{
    keywordstyle=[1]\color{base01}\bfseries,
    keywordstyle=[2]\color{green},
    keywordstyle=[3]\color{blue},
    stringstyle=\color{orange},
    commentstyle=\itshape\color{orange},
    basicstyle=\linespread{1.15}\ttfamilywithbold\color{base00},
    belowskip=0pt,
    belowcaptionskip=1em,
    abovecaptionskip=1em,
    literate = {-}{-}1,
    numbers=left,
    xleftmargin=2em,
    breaklines=true,
}

\renewcommand{\>}{\rangle}
\newcommand{\<}{\langle}

\NewEnviron{abstractsyntax}{%
  \begingroup
  \setlength{\jot}{2cm}
  \vspace{-1.5em}
  \begin{align}
  \begin{split}
    \BODY
  \end{split}
  \end{align}
  \endgroup
}

\NewEnviron{abstractsyntax*}{%
  \begingroup%
  \setlength{\jot}{1mm}%
  \begin{align*}%
    \BODY%
  \end{align*}%
  \endgroup%
}

\NewEnviron{namedabstractsyntax*}{%
  \vspace{-1.5em}
  \begin{alignat*}{2}
    \BODY
  \end{alignat*}
}

\newcommand{\trans}[1][]{\xrightarrow{#1\phantom{}}}
\newcommand{\Trans}[1][]{\xRightarrow{#1\phantom{}}}
\renewcommand{\ss}[1][]{\Trans[#1]}

\renewcommand{\sf}[1]{\textsf{\textup{#1}}}
\newcommand{\rt}[1]{\textsc{(#1)}}

\newcommand{\eqdef}{\triangleq}

\newcommand{\dom}[1]{\text{dom}(#1)}

\def\U{\mathcal{U}}

\newcommand{\substitute}[2]{\ensuremath{\{#1/#2\}}}

\newcommand{\sif}[4][]{\sf{if}_{#1}\ \allowbreak(#2)\ \allowbreak\{#3\}\ \sf{else}\ \allowbreak\{#4\}}

\newcommand{\runtime}[1]{\colorbox{lightgray}{\ensuremath{#1}}}
\newcommand{\ctx}{\ensuremath{\mathcal{E}}}

\newcommand{\uend}{\ensuremath{{\sf{end}}}\xspace}

\newcommand{\case}[1]{\noindent\textbf{Case} #1\textbf{:}}


\makeatletter
\def\widegather{
    \patchcmd\start@gather{$$}{%
      $$%
      \displaywidth=\textwidth
      \displayindent=-\leftskip
    }{}{\errmessage{Cannot patch \string\start@gather}}
}
\makeatother

\newenvironment{rules}
        {\begingroup
        \setlength{\jot}{2mm}%
        \widegather
        \csname gather*\endcsname
    }
    {
    \csname endgather*\endcsname
    \endgroup
    }

\newcommand{\usage}[1]{
    \begingroup
        \ensuremath{#1}%
    \endgroup
}

\usepackage[plain, shortend, linesnumbered, ruled]{algorithm2e}
\SetAlFnt{\sffamily \small}
\SetKwRepeat{Do}{do}{while}
\SetKwProg{Fn}{Function}{:}{end}
\SetKwInOut{Input}{Input}\SetKwInOut{Output}{Output}

\usepackage{subfiles}

\makeatletter
\def\@inferenceFrontName[#1]{%
  \setbox3=\hbox{\normalfont #1}%
  \ifdim \wd3 > \z@
    \unhbox3%
    \hskip\@@nSpace
  \fi
  \@inferenceMiddle
}
\makeatother

\title{Papaya: Global Typestate Analysis of Aliased Objects Extended Version}

\author{Mathias Jakobsen}
\orcid{0000-0002-6128-7004}
\affiliation{%
\institution{University of Glasgow}
\department{School of Computing Science}
\country{United Kingdom}}
\email{m.jakobsen.1@research.gla.ac.uk}

\author{Alice Ravier}
\affiliation{%
\institution{University of Glasgow}
\department{School of Computing Science}
\country{United Kingdom}}
\email{2206245r@student.gla.ac.uk}

\author{Ornela Dardha}
\orcid{0000-0001-9927-7875}
\affiliation{%
\institution{University of Glasgow}
\department{School of Computing Science}
\country{United Kingdom}}
\email{ornela.dardha@glasgow.ac.uk}

\begin{document}
\begin{abstract}
\emph{Typestates} are state machines used in object-oriented programming to specify and verify correct order of method calls on an object. To avoid inconsistent object states, typestates enforce \emph{linear typing}, which eliminates---or at best limits---\emph{aliasing}.  However, aliasing is an important feature in programming, and the state-of-the-art on typestates is too restrictive if we want typestates to be adopted in real-world software systems.

In this paper, we present a type system for an object-oriented language with typestate annotations, which allows for \emph{unrestricted} aliasing, and as opposed to previous approaches it does not require linearity constraints. The typestate analysis is \emph{global} and tracks objects throughout the entire program graph, which ensures that well-typed programs \emph{conform} and \emph{complete} the declared protocols. We implement our framework in the Scala programming language and illustrate our approach using a running example that shows the interplay between typestates and aliases.
\end{abstract}

\maketitle

\section{Introduction}
In class-based object-oriented programming languages, a class defines a number of methods that can be invoked on an object of that class. Often, however, there is an implicit \emph{order} imposed on methods, where some methods should be called before others. For example, a server connection must be opened before sending data, or we might want a clean-up method to be called before freeing resources. These method orderings, or \textit{protocols}, are often defined in varying degrees of formality through documentation or comments, which makes the process difficult and error-prone. Work has been undertaken to include these protocols in the program itself with the introduction of \emph{typestates} for object-oriented languages \cite{Kouzapas2016TypecheckingStMungo, BravettiBehaviouralTypes2020, DeLine2004, Aldrich2009Typestate-orientedProgramming}. 

Common to many of these approaches is that they rely on a \emph{linear} type system, where only a single reference to an object can exist, thus eliminating--or limiting--aliasing.
In Mungo \cite{Kouzapas2016TypecheckingStMungo, BravettiBehaviouralTypes2020} linearity is always enforced, whereas other approaches in languages such as Plaid \cite{Bierhoff2007} and Vault \cite{DeLine2001EnforcingSoftware} allow limited aliasing, while preserving compositionality of the type system such that each class can be type checked in isolation \cite{Fahndrich2002, Militao2010}. These approaches often require programmer annotations to deal with aliasing control or they simply eliminate aliasing altogether.

The difficulty with aliasing in the presence of typestates is that if multiple references exist to a single object, then operations on one object reference can affect the type of multiple other references as well. This is further complicated if we allow aliases to be stored in fields on multiple objects. Consequently, operations on objects of one class impact the well-typedness of other classes, potentially leading to inconsistent objects' states.
Looking at the problem from a more `technical' angle, the difficulty with aliasing in the presence of typestates is due to the discrepancy between the \emph{compositional} nature of typestate-based type systems and the \emph{global} nature of aliasing. To address aliasing one can either (i) allow limited access to an object through aliasing control mechanisms or (ii) if we want unrestricted aliasing then use a form of global analysis. The problem with (i) is that it not trivial to find an alias control mechanism to capture OO programming idioms, and for (ii) while we benefit from the most flexible form of aliasing, the drawback is that we lose compositionality. With the above in mind we pose our research question:

\textbf{RQ}:
\emph{Can we define a typestate-based type system for object-oriented languages that guarantees protocol conformance and completion while allowing unrestricted aliasing?}

In this paper, we answer positively our research question and introduce a global approach to type checking object-oriented programs with typestates, which allows \textit{unrestricted aliasing}, meaning that objects can be freely aliased, and stored in fields of other objects. This is more representative of the sort of aliasing that can occur in real-world programs. In this work we treat typestates in a similar fashion to the line of work on Mungo \cite{Kouzapas2016TypecheckingStMungo, BravettiBehaviouralTypes2020}
and along the same lines, we introduce Papaya, an implementation of a typestate-based type system for Scala.

\paragraph{Contributions} The contributions of this paper are as follows.
\begin{itemize}
\item \textbf{Typestates for Aliased Objects}. We formalise an object-oriented language with typestate annotations.
\begin{itemize}
\item \cref{sec:language} presents the syntax; \cref{sec:typesystem} presents the type system that performs global typestate analysis of unrestricted aliased objects and \cref{sec:semantics} presents the operational semantics. 
\item  \cref{sec:properties} covers the meta-theory of our formalisation and we show that our type system is safe by proving subject reduction (\cref{thm:subject-reduction}), progress (\cref{thm:progress}), protocol conformance (\cref{cor:usage_conformance}) and protocol completion (\cref{lemma:protocol_completion}).
\end{itemize}

\item \textbf{Papaya Tool}. \cref{sec:implementation} presents the Papaya tool, an implementation of our type system for Scala. Protocols are expressed as Scala objects and are added to Scala classes with the \lstinline{@Typestate} annotation. Following the formalisation, our implementation allows for unrestricted aliasing, where objects are checked if they conform and complete their declared protocols.

\item \textbf{The BankAccount Example.} We illustrate our work with a running example (starting in \cref{sec:overview}), which features aliasing. We show how the program is typed in our type system (from \cref{sec:typesystem}) and we implement it in Scala (in \cref{sec:implementation}) where use Papaya to perform typestate checking.
\end{itemize}
In \cref{sec:relatedwork} we discuss related work on typestates and aliasing. Finally, in \cref{sec:conclusion} we conclude the paper and present ideas for future work.

\section{Overview}
\label{sec:overview}

We introduce our approach with an example, which is inspired by \cite{Jakobsen2020}. The example is shown using the calculus that will be defined in \cref{sec:language} with the addition of some base types and operations on those. For completeness, since the calculus requires a formal parameter for all methods, one could pass the \sf{unit} value as an argument. For readability we omit the argument instead.

Consider the class \lstinline{BankAccount} shown in \cref{lst:bankaccount}. It is a simple wrapper class around a field storing an amount of money. Notice that there is an implicit ordering of method calls, which the programmer might assume will be followed when using the class: the amount of money should be set \emph{prior} to using the value of the field, and interest should be applied \emph{after} setting the money; finally, the money variable should only be read \emph{after both} setting the money and applying the interest has occurred, so that an intermediate value is not returned.

\begin{lstlisting}[label={lst:bankaccount}, caption={Wrapper class around an amount of money}]
class BankAccount[{setMoney; 
                  {applyInterest; 
                  {getMoney; end}}}] {
  val amount:float;
  fun setMoney(d:float):void (*@\label{lst:bank-account:setMoney}@*) {
    this.amount = d; 
  }
  fun getMoney():float (*@\label{lst:bank-account:getMoney}@*){
    this.amount;
  }
  fun applyInterest(rate:float) (*@\label{lst:bank-account:applyInterest}@*){
    this.amount = this.amount * rate;
  }
}
\end{lstlisting}

We can express this implicit order of method calls as an explicit \textit{usage}:
\[ 
\usage{\{\texttt{setMoney};\allowbreak\{\texttt{applyInterest};\allowbreak\{\texttt{getMoney};\allowbreak \uend\}\}\}}
\]
where $\{m_i; w_i\}_{i \in I}$ denotes that a method $m_j$ where $j \in I$ can be called, with the continuation usage $w_j$.
This usage states that the first method called should be \texttt{setMoney}, followed by a call to  \texttt{applyInterest} and finally one to \texttt{getMoney}.

We introduce two additional classes as shown in \cref{lst:salary,lst:datastore}. The \lstinline{SalaryManager} class adds money to a \lstinline{BankAccount} and applies a fixed interest rate. The \lstinline{DataStorage} class fetches the value of a \lstinline{BankAccount} and stores it in a database.

\begin{lstlisting}[firstnumber=15, label={lst:salary}, caption={Salary manager that adds funds to a \lstinline{BankAccount} object}]
class SalaryManager[{setAccount; 
                    {addSalary; end}}] {
  val account:BankAccount
  fun setAccount(ms:BankAccount):void {
    this.account = ms;
  }
  fun addSalary(amount:float) {
    this.account.setMoney(amount);
    this.account.applyInterest(1.05);
  }
}
\end{lstlisting}
\begin{lstlisting}[firstnumber=26, label={lst:datastore}, caption={Data storage class that reads the funds of a \lstinline{BankAccount} object}]
class DataStorage[{setAccount; 
                  {store; end}}] {
  val account:BankAccount
  fun setAccount(ms:BankAccount):void {
    this.account = ms;
  }
  fun store() {
    this.account.getMoney();
    // store value in database
  }
}
\end{lstlisting}

Note that in the three classes we defined so far, there is no explicit mentioning of possible aliasing. However, as shown in \cref{lst:aliasing-bankaccount}, an instance of class \lstinline{BankAccount} can be aliased and shared between the manager and data store, as long as the joined operations on the instance respect its usage. 

\begin{lstlisting}[firstnumber=37, label={lst:aliasing-bankaccount}, caption={Aliasing of a \lstinline{BankAccount} object}]
account = new BankAccount;
manager = new SalaryManager;
db = new DataStorage;

manager.setAccount(account);
db.setAccount(account);

manager.addSalary(100.0); (*@\label{lst:aliasing-bankaccount:1}@*)
db.store();(*@\label{lst:aliasing-bankaccount:2}@*)
\end{lstlisting}

If we were to swap lines \ref{lst:aliasing-bankaccount:1} and \ref{lst:aliasing-bankaccount:2}, then they would no longer follow the protocol, as the data store would call \texttt{getMoney} before \texttt{setMoney} and \texttt{applyInterest} were called.

\section{The Language}
\label{sec:language}
We introduce an object-oriented calculus with classes and enumeration types, similar to previous work on Mungo \cite{Kouzapas2016TypecheckingStMungo, Kouzapas2018TypecheckingJava, Dardhaetal2017, BravettiBehaviouralTypes2020, VoineaDG20}.

The syntax of terms is shown in \cref{fig:language}. For a sequence $\phi_1, \phi_2, \ldots \phi_n$ we write $\overline{\phi}$ and let $|\overline{\phi}|=n$. A program is a list of class and enum-definitions $\overline{D}$, followed by a class \sf{Main} which contains the \sf{main} method. A class definition contains the initial protocol, or usage $\U$, field declarations $\overline{F}$ and method declarations $\overline{M}$. For expressions, the language supports assignment, object initialisation, method calls (on fields, parameters or on the object itself).
Note that for simplicity and readability of typing rules later on, method calls and field access use an object-reference $o$ as the target, thus call-chaining and nested field access is not allowed. However, the language can be easily extended to facilitate these features, requiring an extra object look-up in the relevant typing rules. The only object reference that can occur in program text is the $\sf{this}$ reference.
The language also supports control structures (conditionals, loops, and sequential composition) and match expressions (switch on an enumeration type). Loops are formalised with a jump-style loop with labelled expressions and \sf{continue} statements in line with Mungo work. 

\begin{figure}[htpb]
    \centering
    \begin{subfigure}[b]{0.96\columnwidth}
        \begin{abstractsyntax*}
            D ::=&\ \sf{class}\ C \{\U, \overline{F}, \overline{M}\} \mid \sf{enum}\ L \{\overline{l}\} \\
            F ::=&\ \sf{val}\ f : t \\
            M ::=&\ \sf{fun}\ m(x : t) : t\ \{e\} \\
            r ::=&\ o \mid o.f \mid x\\
            e ::=&\ o.f = e \mid o.f = \sf{new}\ C \mid e;e \mid r.m(e) \mid \sf{unit} \mid o.f \mid x \\ 
            \mid&\ \sif{e}{e}{e} \mid o.l \mid \sf{match} (e) \{ \overline{l : e}\} \mid \sf{null}\\ 
            \mid&\ \sf{true} \mid \sf{false} \mid k : e \mid \sf{continue}\ k 
        \end{abstractsyntax*}
        \caption{Syntax of class definitions}
        \label{fig:language}
    \end{subfigure}
    \begin{subfigure}[b]{0.96\columnwidth}
        \begin{abstractsyntax*}
            t ::=&\ C \mid \sf{void} \mid \sf{bool} \mid L\\
            \runtime{T} ::=&\ \runtime{o[C, \U]} \mid \sf{void} \mid \runtime{\bot} \mid \sf{bool} \mid L \mid \runtime{L\ \sf{link}\ o}\\
            \U ::=&\ \mu X.\U \mid X \mid \{ \overline{m; w} \} \mid \uend \\
            w ::&\ \<\overline{l: \U}\> \mid \U
        \end{abstractsyntax*}
        \caption{Syntax of types}
        \label{fig:types}
    \end{subfigure}
    \caption{Syntax of terms and types}
    \label{fig:syntax}
\end{figure}

The syntax of types is shown in \cref{fig:types} and it contains the object types $o[C, \U]$, base types \sf{bool} and \sf{void}, the \textit{null-type} $\bot$, and enumeration types $L$ and $L\ \sf{link}\ o$. The shaded production rules indicate run-time syntax. An object type $o[C, \U]$ is composed of an object reference $o$, which is a unique identifier a single object, a class name $C$ and a current usage $\U$ describing the remaining protocol of the object. The enumeration type $L\ \sf{link}\ o$ introduced in \cite{VasconcelosGay2009DynamicInterfaces} is used to track updates in switch-statements and are not declared in the program text.

\cref{def:lts} presents a labelled transition system for usages, annotated with the method call or the enumeration label, depending on the action performed. If an object has type $o[C, \U]$, then the transitions of $\U$ describe the permitted operations on the object referenced by $o$. As previously described, branch usages $\{m_i; w_i\}_{i \in I}$ describe a set of available methods, each with a continutation usage. Choice usages $\<l_i : \U_i\>_{i \in I}$ describe that based on a enumeration label $l_j$, the protocol continues with protocol $\U_j$. Recursive behaviour can be specified with recursive usages $\mu X.\U$ and the \uend usage denotes the terminated protocol which has no transitions.

\begin{definition}[LTS on Usages]
\label{def:lts}
\begin{gather*}
    \inferrule{j \in I}{\{m_i; w_i\}_{i \in I} \trans[m_j] w_j} \;
    \inferrule{j \in I}{\<l_i : \U_i\>_{i \in I} \trans[l_j] \U_j} \\ 
    \inferrule{\U\substitute{X}{\mu X.\U} \trans \U'}{\mu X.\U \trans[\alpha] \U'}
\end{gather*}
\end{definition}

We define a notion of well-formedness for expressions (\cref{def:well-formedness}), which requires that \sf{continue} expressions do not show up in places where, after loop unfolding, they would be followed by other expressions. Examples of ill-formed expressions include $o.m(\sf{continue}\ k)$ and $\sf{continue}\ k;o.m(\sf{unit})$. Furthermore, well-form\-ed\-ness also requires a labelled expression has a terminating branch so that $k : \sif{\sf{true}}{\sf{continue}\ k}{\sf{unit}}$ is well-formed whereas $k : \sf{continue}\ \allowbreak k$ is not.

\begin{definition}[Well-formedness]
\label{def:well-formedness}
An expression $e$ is well-formed if:

\begin{enumerate}
    \item No expression follows a \sf{continue} expression after unfolding \sf{continue} expressions in $e$
    \item No free loop-variables in $e$
    \item All \sf{continue} expressions in $e$ are guarded by a branching (\sf{if} or \sf{match}) expression
    \item There must be a branch in all labelled expressions in $e$ that does not end with a \sf{continue} expression
\end{enumerate}
\end{definition}

We conclude with the definition of well-formed methods.

\begin{definition}[Well-formed methods]
\label{def:well-formed methods}
A method declaration $\sf{fun}\ \allowbreak m(x : t) : t\ \{e\}$ is well formed if $e$ is well formed and recursive calls are guarded by a branching expression.
\end{definition}

\section{Type System}
\label{sec:typesystem}

As opposed to previous type systems for Mungo \cite{Kouzapas2016TypecheckingStMungo, Kouzapas2018TypecheckingJava, BravettiBehaviouralTypes2020} the type system presented here performs a \emph{global} analysis of the program, in order to maintain a global view of aliasing while guaranteeing correct objects' states. This means that instead of relying on compositionality during type checking, we must explore the entire program graph. Consequently when a method call is encountered during type checking, the type system must ensure that the body of the method is well typed in the current typing environment.

We define a typing environment $\Gamma$ using the production rules shown in \cref{fig:typingenvironment}. A typing environment maps object references to an object-type as well as a field typing environment $\lambda$ that contains the types for all fields in the object. Furthermore we use the notation $\Gamma[o \mapsto (T, \lambda)]$ to indicate an update of an existing binding for object $o$, and $\Gamma[o.f \mapsto o']$ to update the existing binding of a field of object $o$. A typing environment can only contain a single binding for each object reference $o$. Similarly, a field typing environment can only contain a single binding for each field name.

\begin{figure}[htpb]
    \centering
    \begin{abstractsyntax*}
        \Gamma ::=\ &\emptyset \mid \Gamma, o \mapsto (T, \lambda) \\
        \lambda ::=\ &\emptyset \mid \lambda, f \mapsto z \\
        z ::=\ &\sf{basetype bool} \mid \sf{basetype void}\\ \mid\ &\sf{basetype}\ \bot \mid\ \sf{basetype}\ L \\ \mid\ &\sf{reference}\ o
    \end{abstractsyntax*}
    \caption{Syntax of typing environments}
    \label{fig:typingenvironment}
\end{figure}

We define the initial field environment given a set of field declarations $\overline{F}.\sf{inittypes}$. This is used when initialising new objects. Fields with class types are given the initial type of $\bot$ whereas fields of base types retain that type in the field environment.

\begin{align*}
    (\overline{F}, \sf{var}\ f : C).\sf{inittypes} &= \overline{F}.\sf{inittypes}, f \mapsto \sf{basetype}\ \bot \\
    (\overline{F}, \sf{var}\ f : \sf{bool}).\sf{inittypes} &= \overline{F}.\sf{inittypes}, f \mapsto \sf{basetype}\ \sf{bool} \\
    (\overline{F}, \sf{var}\ f : \sf{void}).\sf{inittypes} &= \overline{F}.\sf{inittypes}, f \mapsto \sf{basetype}\ \sf{void} \\ 
    (\overline{F}, \sf{var}\ f : L).\sf{inittypes} &= \overline{F}.\sf{inittypes}, f \mapsto \sf{basetype}\ L \\ 
    \emptyset.\sf{inittypes} &= \emptyset
\end{align*}

We also define the following shorthand functions for extracting information from the typing environment and class definitions.

\begin{equation*}
    \begin{aligned}[c]
        (o[C, \U], \lambda).\sf{class} &\eqdef C \\
        (o[C, \U], \lambda).\sf{usage} &\eqdef \U \\
        (o[C, \U], \lambda).\sf{reference} &\eqdef o \\
    \end{aligned}
    \qquad
    \begin{aligned}[c]
        (o[C, \U], \lambda).f &\eqdef \lambda(f) \\
        (T, \lambda).\sf{type} &\eqdef T \\
        (T, \lambda).\sf{fields} &\eqdef \lambda
    \end{aligned}
\end{equation*}

For a class name $C$ where $\sf{class}\ C\{\U, \overline{F}, \overline{M}\}\in \overline{D}$ we let $\overline{D}(C)=\sf{class}\ C\{\U, \overline{F}, \overline{M}\}$ and define the following functions.

\begin{align*}
    (\sf{class}\ C\{\U, \overline{F}, \overline{M}).\sf{usage} &\eqdef\U \\
    (\sf{class}\ C\{\U, \overline{F}, \overline{M}).\sf{fields} &\eqdef \overline{F} \\
    (\sf{class}\ C\{\U, \overline{F}, \overline{M}).\sf{methods} &\eqdef \overline{M}
\end{align*}

The type system is driven by the following (Main) rule, which states that if the main method is well typed, then the entire program is well typed. As previously mentioned, the type system will expand method calls, hence the type system will visit all reachable parts of the program. In the (Main) rule we require $\sf{term}(\Gamma)$ meaning that the resulting type environment must be \textit{terminated}, meaning that protocols must be finished for all objects. \sf{term} is defined as:
\[
    \sf{term}(\Gamma) \Leftrightarrow \forall o \in \text{dom}(\Gamma).\ \Gamma(o).\sf{usage} = \uend
\]

\begin{gather*}
    \inferrule[Main]{\sf{Main} \{\U, \overline{F}, \overline{M}\} \in \overline{D} \\ \overline{M} = \{\sf{fun}\ \sf{main}(\sf{void}\ x)\ \{e\}\} \\ \U = \{\sf{main}; \uend\} \\ \emptyset;\emptyset;\{o_{\sf{main}} \mapsto (\sf{Main}[\uend], \overline{F}.\sf{initvals})\} \vdash e: T \dashv \Gamma\\\sf{term}(\Gamma)}{\vdash \overline{D} : \sf{ok}}
\end{gather*}

Judgments for type checking expressions are of the form $\Theta;\Omega;\Gamma \vdash e : T \dashv \Gamma'$. The environments $\Omega$ and $\Theta$ are used to track labelled expressions and recursive method calls respectively. The label environment $\Omega$ relates loop labels $k$ to typing environments $\Gamma$ such that when encountering a \sf{continue} expression we can compare the current typing environment to the initial typing environment when entering the loop. This will be explained in detail later. The recursion environment $\Theta$ serves the same purpose but for recursive method calls instead. As method calls are expanded in the type systems, recursive method definitions will lead to infinite type checking if not handled carefully. By keeping track of the currently expanded methods, the type system can terminate type checking after a single expansion of each method.

Returning back to the format of judgments, $\Theta;\Omega;\Gamma \vdash e : T \dashv \Gamma'$, we can now describe the meaning of the judgment. Given initial environments $\Theta$, $\Omega$, $\Gamma$, evaluating the expression $e$ will result in a value of type $T$ and a possibly updated typing environment $\Gamma'$. We assume that $\overline{D}$ is globally available in the rules, and contains the class definitions of the program.

The first set of rules, found in \cref{fig:typing-objects}, describes object operations such as reading fields or parameters, assigning fields and object initialisation. They are for the most part standard, although the rules for method call require some further description.

The rules for direct method calls, which is used for typing method calls on parameters or the \sf{this} object, are defined in (Call-d) and (Call-d-rec). In (Call-d) the method body is expanded, and the current typing environment from before the unfolding is stored in the recursion environment $\Theta$. Upon reaching a recursive call inside the method body the (Call-d-rec) rule is used to compare the current typing environment to the one stored in $\Theta$. This enforces that upon making a recursive call, the typing environment should be the same as it was for the initial method call. If this is the case, then we can terminate type checking of the call-chain, as we have checked this exact configuration already with the initial call. So the combination of the two rules (Call-d) and (Call-d-rec) gives us a recursive typing behaviour, with the base case defined by (Call-d-rec). This exact behaviour is repeated for indirect calls which are used for fields, as illustrated in the rules (Call-ind) and (Call-ind-rec).

Four auxiliary functions are used in the rules:

\sf{agree} checks that a value of type $T$ matches the one defined in the program text as $t$. This allows \sf{null} to be written to fields with class types, and allows objects to be stored in fields with matching classes, no matter the particular protocol of the object.
\begin{gather*}
    \sf{agree}(C, \bot) \quad
    \sf{agree}(C, o[C, \U]) \quad
    \sf{agree}(\sf{bool}, \sf{bool}) \\
    \sf{agree}(\sf{void}, \sf{void}) \quad
    \sf{agree}(L, L)
\end{gather*}

The \sf{returns} predicate extends the \sf{agree} predicate with an option to return a link type from a method, to support switching on choice usages by linking the enumeration type to an object.

\[
    \sf{returns}(t, T) \Leftrightarrow \sf{agree}(t, T) \vee (t = L \wedge T = L\ \sf{link}\ o)
\]

\sf{getType} and \sf{vtype} are used for tagging and unpacking values for storing them in the typing environment. The reason we need this is to handle the indirection of an object reference $o$, so that we can look up the type of a field, with an extra access to the typing environment.

\begin{align*}
    \sf{getType}(\sf{reference}\ o, \Gamma) &= \Gamma(o).\sf{type} \\
    \sf{getType}(\sf{basetype}\ \sf{bool}, \Gamma) &= \sf{bool}\\
    \sf{getType}(\sf{basetype}\ \sf{void}, \Gamma) &= \sf{void}\\
    \sf{getType}(\sf{basetype}\ \bot, \Gamma) &= \bot\\
    \sf{getType}(\sf{basetype}\ L, \Gamma) &= L
\end{align*}

\begin{align*}
    \sf{vtype}(o[C, \U]) &= \sf{reference}\ o \\
    \sf{vtype}(\sf{bool}) &= \sf{basetype}\ \sf{bool} \\
    \sf{vtype}(\sf{void}) &= \sf{basetype}\ \sf{void} \\
    \sf{vtype}(\bot) &= \sf{basetype}\ \bot \\
    \sf{vtype}(L) &= \sf{basetype}\ L \\
\end{align*}

\begin{figure*}[htbp]
\begin{subfigure}[b]{0.96\textwidth}
    \hfill \fbox{$\Theta; \Omega; \Gamma \vdash e : T \dashv \Gamma'$}
    \begin{rules}
        \inferrule[Assign]{\Theta;\Omega;\Gamma \vdash e : T \dashv \Gamma'  \\ \Gamma'(o).\sf{class}.\sf{fields}(f) = \sf{var}\ f: t \\ \sf{agree}(t, T)}{\Theta;\Omega;\Gamma \vdash o.f = e : \sf{void} \dashv \Gamma'[o.f \mapsto \sf{vtype}(T)]} \quad
        \inferrule[Field]{\Gamma(o).\sf{fields}(f) = z \\ T = \sf{getType}(\Gamma, z)}{\Theta;\Omega;\Gamma \vdash o.f : T \dashv \Gamma} \\
        \inferrule[New]{o'\ \text{fresh}  \\ \overline{D}(C)=\sf{class}\ C \{\U, \overline{F}, \overline{M}\} \\ \sf{val}\ f : C \in \overline{D}(\Gamma'(o).\sf{class}).\sf{fields}}{\Theta;\Omega;\Gamma \vdash o.f = \sf{new}\ C : \sf{unit} \dashv (\Gamma, o' \mapsto (o'[C, \U], \overline{F}.\sf{inittypes}))[o.f\mapsto o']}  \quad
        \inferrule[Unit]{ }{\Theta;\Omega;\Gamma \vdash \sf{unit} : \sf{void} \dashv \Gamma} \\
        \inferrule[Bool]{v \in \{\sf{true}, \sf{false}\}}{\Theta;\Omega;\Gamma \vdash v : \sf{bool} \dashv \Gamma} \quad 
        \inferrule[Enum]{l \in L}{\Theta;\Omega;\Gamma \vdash o.l : L\ \sf{link}\ o\dashv \Gamma} \quad
        \inferrule[Null]{ }{\Theta;\Omega;\Gamma \vdash \sf{null} : \bot \dashv \Gamma} \quad
        \inferrule[Const]{l \in L}{\Theta;\Omega;\Gamma \vdash o.l : L \dashv \Gamma} \quad 
        \inferrule[Obj]{\Gamma(o) = (o[C, \U], \lambda)}{\Theta;\Omega;\Gamma \vdash o : o[C, \U] \dashv \Gamma}\\
        \inferrule[Call-d]{\Theta;\Omega;\Gamma \vdash e : T \dashv \Gamma'' \\ \Gamma''(o) = (o[C, \U], \lambda) \\ \U \trans[m] \U' \\\\ \sf{fun}\ m(x : t) : t' \{e'\} \in \overline{D}(C).\sf{methods} \\ \sf{agree}(t, T)\\\\
       (\Theta, o.m \mapsto \Gamma''');\Omega;\Gamma''' \vdash e'\substitute{\sf{this}}{o}\substitute{x}{\sf{getValue}(T')} : T' \dashv \Gamma' \\ \sf{returns}(t', T')}{\Theta;\Omega;\Gamma \vdash o.m(e) : T' \dashv \Gamma'} \\ \hspace{5.5cm}\text{where $\Gamma'''=\Gamma''[o\mapsto(o[C, \U'], \lambda)$}\\
       \inferrule[Call-d-rec]{(\Theta, o.m \mapsto \Gamma'');\Omega;\Gamma \vdash e : T \dashv \Gamma''' \\ \Gamma'''(o) = (o[C, \U], \lambda) \\ \U \trans[m] \U' \\\\ \sf{agree}(t, T)\\ \Gamma'' = \Gamma'''[o \mapsto (o[C, \U'], \lambda)]}{(\Theta, o.m \mapsto \Gamma'');\Omega;\Gamma \vdash o.m(e) : T' \dashv \Gamma'}\\
        \inferrule[Call-ind]{\Theta;\Omega;\Gamma \vdash e : T \dashv \Gamma'' \\ \Gamma''(o).f = o' \\\\ \Gamma''(o') = (o'[C, \U], \lambda) \\ \U \trans[m] \U' \\\\ \sf{fun}\ m(x : t) : t' \{e'\} \in \overline{D}(C).\sf{methods} \\ \sf{agree}(t, T)\\\\
        (\Theta, o'.m \mapsto \Gamma''');\Omega;\Gamma''' \vdash e'\substitute{\sf{this}}{o}\substitute{x}{\sf{getValue}(T')} : T' \dashv \Gamma' \\ \sf{returns}(t', T')}{\Theta;\Omega;\Gamma \vdash o.f.m(e) : T' \dashv \Gamma'}  \\\hspace{5.5cm}\text{where $\Gamma'''=\Gamma''[o'\mapsto(o'[C, \U], \lambda)$} \\
        \inferrule[Call-ind-rec]{(\Theta, o'.m \mapsto \Gamma'');\Omega;\Gamma \vdash e : T \dashv \Gamma''' \\ \Gamma'''(o).f = o' \\ \Gamma'''(o') = (o'[C, \U], \lambda) \\\\ \U \trans[m] \U' \\ \sf{agree}(t, T) \\ \Gamma'' = \Gamma'''[o' \mapsto (o'[C, \U'], \lambda)]}{(\Theta, o'.m \mapsto \Gamma'');\Omega;\Gamma \vdash o.f.m(e) : T' \dashv \Gamma'}\\
    \end{rules}
    \caption{Object operations and values}
    \label{fig:typing-objects}
\end{subfigure}
\begin{subfigure}[b]{0.96\textwidth}
    \hfill \fbox{$\Theta; \Omega; \Gamma \vdash e : T \dashv \Gamma'$}
    \begin{rules}
        \inferrule[If]{\Theta;\Omega;\Gamma \vdash e_1 : \sf{bool} \dashv \Gamma''\\ \Theta;\Omega;\Gamma'' \vdash e_2 : T \dashv \Gamma'\\ \Theta;\Omega;\Gamma'' \vdash e_3 : T \dashv \Gamma'}{\Theta;\Omega;\Gamma \vdash \sif{e_1}{e_2}{e_3} : T \dashv \Gamma'} \\
        \inferrule[Comp]{\Theta;\Omega;\Gamma \vdash e : T \dashv \Gamma'' \\ \Theta;\Omega;\Gamma'' \vdash e' : T'\dashv \Gamma'}{\Theta;\Omega;\Gamma \vdash e;e' : T' \dashv \Gamma'} \quad
        \inferrule[Label]{\Theta;\Omega, k \mapsto \Gamma;\Gamma \vdash e : T \dashv \Gamma'}{\Theta;\Omega;\Gamma \vdash k : e : T\dashv \Gamma'} \quad
        \inferrule[Continue]{\Omega(k) = \Gamma}{\Theta;\Omega;\Gamma \vdash \sf{continue}\ k : T \dashv \Gamma'} \\
        \inferrule[Case]{\Theta;\Omega;\Gamma \vdash e : L\ \sf{link}\ o \dashv \Gamma'' \\ \forall l_i\in L. 
        {\begin{cases}
            \Gamma''(o).\sf{usage} \trans[l_i] \U_i \\ \Theta;\Omega;\Gamma''[o.\sf{usage} \mapsto \U_i]\vdash e_i : T \dashv \Gamma'
        \end{cases}}}{\Theta;\Omega;\Gamma \vdash \sf{match}(e)\{\overline{l: e}\} : T \dashv \Gamma'}
    \end{rules}
    \caption{Composite expressions}
    \label{fig:typing-composite}
\end{subfigure}
    \caption{Typing rules for expressions}
    \label{fig:typing-rules}
\end{figure*}

Next follows the rules for control structures shown in \cref{fig:typing-composite}. The rule for sequential composition is straightforward, and uses the resulting environment from $e$ to type check $e'$. The rule (If) requires the same resulting environment for both branches. The (Case) rule checks that all branches of the choice-usage results in the same final environment, similar to (If) and updates the usage accordingly. Finally (Label) and (Continue) type checks loops. (Label) simply adds the current typing environment to $\Omega$, while (Continue) can result in an arbitrary type and typing environment. This behaviour is safe, since we know that the expression is well formed, meaning that any \sf{continue} statement is guarded by an if case or match statement, in which case only the choice of an environment that matches the other branch, can be chosen if the expression must be well typed.

Notice that in the (Case) rule we make use of the enumeration type $L\ \sf{link}\ o$. As such types do not \textit{agree} with any other types, they cannot be stored in fields or used as method arguments. So the only way for these types to be show up in a well-typed program, is as the return value of a method, used for matching in a case statement.

\begin{example}
\label{ex:typesystem_bank}

Consider again the bank account example presented in \cref{sec:overview}. The typing environments is an approximation of the heap, and after typing lines 41-46, the environment is:

\begin{align*}
    o_{\sf{main}} &\mapsto \begin{array}[t]{l}(\sf{Main}[\uend],\\ \{\texttt{account} \mapsto o_{\sf{acc}}, \texttt{manager} \mapsto o_{\sf{man}}, \texttt{db} \mapsto o_{\sf{d}}, \})\end{array} \\
    o_{\sf{acc}} &\mapsto \begin{array}[t]{l}(\sf{BankAccount}[\begin{array}[t]{l}\{\sf{setMoney}; \{\sf{applyInterest}; \\\{ \sf{getMoney}; \uend\}\}\}],\end{array}\\ \{\texttt{amount} \mapsto \texttt{double}\})\end{array} \\
    o_{\sf{man}} &\mapsto (\sf{SalaryManager}[\{\sf{addSalary}; \uend\}], \{\texttt{account} \mapsto o_{\sf{a}}\})\\
    o_{\sf{db}} &\mapsto (\sf{DataStorage}[\{\sf{store}; \uend\}], \{\texttt{account} \mapsto o_{\sf{a}}\})\\
\end{align*}

After type checking line 48, where the salary manager adds funds to the account, the following bindings are updated in the typing environment, while the remaining bindings are unchanged.

\begin{align*}
    o_{\sf{acc}} &\mapsto (\sf{BankAccount}[\{ \sf{getMoney}; \uend\}],\{\texttt{amount} \mapsto \texttt{double}\})\\
    o_{\sf{man}} &\mapsto (\sf{SalaryManager}[\{\sf{addSalary}; \uend\}], \{\texttt{account} \mapsto o_{\sf{a}}\})
\end{align*}

This allows line 49 to be typechecked, since the $o_{\sf{acc}}$ has been updated to allow a call to the method \texttt{getMoney}, this is an example of how the global type checking approach allow us to track changes to aliased fields, even if they happen through seemingly unrelated objects.

\end{example}

\section{Semantics}
\label{sec:semantics}

In this section we define the run-time semantics of the language. It follows the standard model where object references are used to look up values in the heap. The heap itself is similar in some respects to the typing environment we have previously discussed. The heap maps object references $o$ to their class and a field environment. In the semantics, we do not consider typestates, hence instead of mapping the object reference to a full type, we only map it to its class in order to look up method definitions and field declarations. The field bindings in the heap is a mapping from field names to values, which themselves can be object references or base values such as \sf{true}, \sf{null}, or \sf{unit}.
The initial field environment is defined similarly to $\overline{F}.\sf{inittypes}$, but instead it maps the fields to values instead, in $\overline{F}.\sf{initvals}$.

\begin{align*}
    (\overline{F}, \sf{var}\ f : C).\sf{initvals} &= \overline{F}.\sf{initvals}, f \mapsto \sf{null} \\
    (\overline{F}, \sf{var}\ f : \sf{bool}).\sf{initvals} &= \overline{F}.\sf{initvals}, f \mapsto \sf{false} \\
    (\overline{F}, \sf{var}\ f : \sf{void}).\sf{initvals} &= \overline{F}.\sf{initvals}, f \mapsto \sf{unit} \\ 
    (\overline{F}, \sf{var}\ f : L).\sf{initvals} &= \overline{F}.\sf{initvals}, f \mapsto l\ \\&\text{where}\ \sf{enum}\ L \{l, \overline{l}\} \in \overline{D} \\ 
    \emptyset.\sf{initvals} &= \emptyset
\end{align*}

We now define configurations, which are of the form $\<h, e\>$. When evaluating the expression $e$, both the expression and the heap can change. To model this, we let a computation step be of the form $\<h, e\> \ss \<h', e'\>$. The reduction rules are shown in \cref{fig:semantics}.

To simplify the reduction rules, we make use of an evaluation context to guide the evaluation of composite expressions.

\begin{abstractsyntax*}
    \ctx ::=&\ [\_] \mid o.f =\ctx \mid  \ctx ; e \mid o.m(\ctx) \mid o.f.m(\ctx) \\ \mid &\ \sif{\ctx}{e}{e} \mid  \sf{match}(\ctx)\{\overline{l : e}\}
\end{abstractsyntax*}

The \rt{ctx} rules ensures that inner expressions are evaluated first (e.g. left-hand side of a sequential expression are evaluated before right-hand side). The remaining rules handle the interesting base cases of the semantics.

\begin{figure*}[htbp]
\centering
\hfill \fbox{$\<h, e\> \ss \<h', e'\>$} \hspace{3cm}
\begin{rules}
    \inferrule[ctx]{\<h, e\> \ss \<h', e'\>}{\<h, \ctx[e]\>\ss \<h', \ctx[e']\>} \quad
    \inferrule[assign]{ }{\<h, o.f = v\> \ss \<h[o.f \mapsto v], \sf{unit}\>} \quad
    \inferrule[seq]{ }{\<h, v;e\> \ss \<h, e\>} \\
    \inferrule[if-true]{ }{\<h, \sif{\sf{true}}{e_1}{e_2}\> \ss \<h, e_1\>} \quad
    \inferrule[if-false]{ }{\<h, \sif{\sf{false}}{e_1}{e_2}\> \ss \<h, e_2\>} \\
    \inferrule[lab]{ }{\<h, k : e\> \ss \<h, e\substitute{\sf{continue}\ k}{k : e}\>} \quad
    \inferrule[match]{l_j: e_j \in \overline{l: e}}{\<h, \sf{match}(o.l_j) \{\overline{l : e}\}\> \trans \<h, e_j\>} \\
    \inferrule[call-d]{h(o).\sf{class}.\sf{methods} \ni \sf{fun}\ m(x : t) : t' \{ e \}}{\<h, o.m(v)\> \ss \<h, e\substitute{\sf{this}}{o}\substitute{x}{v}\>} \quad
    \inferrule[call-ind]{h(o).f = o' \\ h(o').\sf{class}.\sf{methods} \ni \sf{fun}\ m(x : t) : t' \{ e \}}{\<h, o.f.m(v)\> \ss \<h, e\substitute{\sf{this}}{o'}\substitute{x}{v}\>} \\
    \inferrule[new]{o'\ \text{fresh} \\ h' = (h, o' \mapsto (C, C.\sf{fields}.\sf{initvals}))[o.f \mapsto o']}{\<h, o.f = \sf{new}\ C\> \ss \<h', \sf{unit}\>} \quad
    \inferrule[fld]{ h(o).\sf{fields}(f)=v}{\<h, o.f\> \ss \<h, v\>}
\end{rules}
\caption{Run-time semantics}
\label{fig:semantics}
\end{figure*}

\begin{example}
\label{ex:semantics_bank}

Consider again the bank account example. When reaching line 48, the heap contains the following bindings:

\begin{align*}
    o_{\sf{main}} &\mapsto (\sf{Main},  \{\texttt{account} \mapsto o_{\sf{acc}}, \texttt{manager} \mapsto o_{\sf{man}}, \texttt{db} \mapsto o_{\sf{d}}, \}) \\
    o_{\sf{acc}} &\mapsto (\sf{BankAccount}, \{\texttt{amount} \mapsto 0\}) \\
    o_{\sf{man}} &\mapsto (\sf{SalaryManager}, \{\texttt{account} \mapsto o_{\sf{a}}\})\\
    o_{\sf{db}} &\mapsto (\sf{DataStorage}, \{\texttt{account} \mapsto o_{\sf{a}}\})\\
\end{align*}

After evaluating the expression on line 48, where the salary manager adds funds to the account, the following bindings are updated in the heap, while the remaining bindings are unchanged.

\[
    o_{\sf{acc}} \mapsto (\sf{BankAccount},\{\texttt{amount} \mapsto 100\})
\]

We see that compared to the type system, fewer bindings were updated, due to the typestates not being tracked in the semantics. However, the resulting environments from the type system and the semantics remains consistent, meaning that the types mentioned in the type system are consistent with the values in the heap. This property and more will be shown in the following section.

\end{example}

\section{Properties}
\label{sec:properties}
In this section we show important properties that hold for the defined language. The first result we show is the fact that we can remove bindings from $\Theta$ while the expression remains well-typed. The intuition of this is that $\Theta$ serves to denote the base case of checking recursive calls. So when we remove a binding from the environment, we simply have to expand the method body once more, leading to the entry being added again in $\Theta$. 

\def\typingfinal{\Gamma^{F}}
\def\typingbound{\Gamma^{N}}
\def\exprbody{e_b}
\def\exprparorig{e_{borig}}
\def\exprorig{o.m(\exprparorig)}
\def\typeorig{T_{orig}}
\def\typingorig{\Gamma_{orig}}

\begin{restatable}{lemma}{recursionlemma}
\label{lemma:recursion_unfolding}
If in a typing derivation starting from an empty recursion environment we have $\Theta, o.m \mapsto \typingbound ; \emptyset ; \Gamma \vdash e : T \dashv \typingfinal$ then we also have $\Theta ; \emptyset ; \Gamma \vdash e : T \dashv \typingfinal$.
\end{restatable}
\begin{proof}
Details in \cref{app:recursion_unfolding_proof}.
\end{proof}

Along with a similar proof for labelled expressions, where bindings can be removed from $\Omega$, this shows that we often consider situations where $\Theta$ and $\Omega$ are empty. So for readability of the upcoming properties, we omit writing the environments when they are empty, so $\Gamma \vdash e : T \dashv \Gamma'$ is equivalent to $\emptyset;\emptyset;\Gamma \vdash e : T \dashv \Gamma'$.

We must establish a soundness result, and show that the usages defined for classes are respected at run-time, and no protocol deviation occurs. To establish such a relationship between the type system and the semantics, we first define \textit{consistency} between a heap and a typing environment, which describes that the typing environment and heap agree on the classes of all objects, and agree on the field bindings of all objects. 

\begin{definition}[Heap consistency]
We say that a heap is consistent with a typing environment, written $\Gamma \vdash h$, if $\Gamma$ and $h$ contains the same objects and the field bindings of each object are also consistent. 

\[
    \inferrule{\dom{h}=\dom{\Gamma} \\ \forall o \in \dom{\Gamma}.h(o).\sf{fields}=\Gamma(o).\sf{fields}\wedge h(o).\sf{class}=\Gamma(o).\sf{class}}{\Gamma \vdash h}
\]

\end{definition}
Furthermore, we lift the transition system for usages to typing environments, with the rules shown in \cref{fig:trans_typing_context}. Notice how the transitions match the updates to a typing environment performed by the typing rules shown in \cref{fig:typing-objects}. This allows us to establish that only a single update is performed to a typing environment when evaluating one step in the reduction semantics.

\begin{figure}
    \centering
    \begin{gather*}
        \inferrule[empty]{ }{\Gamma \trans[\varepsilon] \Gamma} \quad
        \inferrule[trans]{\Gamma(o).\sf{usage} \trans[\alpha] \U}{\Gamma \trans[o.\alpha] \Gamma[o.\sf{usage} \mapsto \U]} \quad 
        \inferrule[update]{\Gamma(o).f = t}{\Gamma \trans[\varepsilon] \Gamma[o.f \mapsto t']} \\
        \inferrule[new]{o \in \text{dom}(\Gamma) \\ o'\ \text{fresh} \\ \sf{class}\ C \{\U, \overline{F}, \overline{M}\} \in \overline{D}}{\Gamma \trans[\varepsilon] (\Gamma, o' \mapsto (o'[C, \U], \overline{F}.\sf{inittypes}))[o.f\mapsto o']}
    \end{gather*}
    \caption{Transition system for typing environments}
    \label{fig:trans_typing_context}
\end{figure}

To complete subject reduction (\cref{thm:subject-reduction}), we consider a semantics where the transitions in (call-d) and (call-ind) are annotated with $o.m$ and $o'.m$ respectively, (match) is annotated with $o.l'$ and all other transitions are annotated with the empty string $\varepsilon$. We can use these labels to show a correspondence between the transitions on typing environments, and the transitions between run-time configurations.

The subject reduction theorem states that a single reduction of a well-typed expression can be matched by a single transition from a consistent typing environment. In other words, this tells us that a single reduction preserves well-typedness with a single update to the typing environment.

\begin{restatable}[Subject Reduction]{theorem}{subjectreduction}
\label{thm:subject-reduction}
If $\Gamma \vdash h$, $\Gamma \vdash e:T \dashv \Gamma'$ and $\<h, e\> \ss[\alpha] \<h', e'\>$ then $\exists \Gamma''.\Gamma'' \vdash e' : T \dashv \Gamma'$ such that $\Gamma \trans[\alpha] \Gamma''$ and $\Gamma'' \vdash h'$
\end{restatable}

\begin{proof}
See details in \cref{app:proof-subject-reduction}
\end{proof}

\begin{example}
\label{ex:properties_bank}
We return to the configuration just before executing line \ref{lst:aliasing-bankaccount:1} in \cref{lst:aliasing-bankaccount}. In \cref{ex:typesystem_bank,ex:semantics_bank} we have stated what the heap and typing environment contain when reaching this statement. In this example we show the consistency between the heap and typing environment for the expression after a single transition.

The remaining expression of the program at this point is:

\[
    e = o_{\sf{main}}.\texttt{manager}.\texttt{addSalary}(100.0);o_{\sf{main}}.\texttt{db}.\texttt{store}(\texttt{unit})
\]

In \cref{ex:typesystem_bank} we identified the typing environment as: 
\begin{align*}
\Gamma = \left\{\begin{array}{l}
    o_{\sf{main}} \mapsto \begin{array}[t]{l}(\sf{Main}[\uend],\\ \begin{array}[t]{l}\{\texttt{account} \mapsto o_{\sf{acc}}, \texttt{manager} \mapsto o_{\sf{man}},\\ \texttt{db} \mapsto o_{\sf{d}}, \})\end{array}\end{array} \\
    o_{\sf{acc}} \mapsto \begin{array}[t]{l}(\sf{BankAccount}[\begin{array}[t]{l}\{\sf{setMoney};\\ \{\sf{applyInterest}; \\\{ \sf{getMoney}; \uend\}\}\}],\end{array}\\ \{\texttt{amount} \mapsto \texttt{double}\})\end{array} \\
    o_{\sf{man}} \mapsto \begin{array}[t]{l}(\sf{SalaryManager}[\{\sf{addSalary}; \uend\}],\\ \{\texttt{account} \mapsto o_{\sf{a}}\})\end{array}\\
    o_{\sf{db}} \mapsto \begin{array}[t]{l}(\sf{DataStorage}[\{\sf{store}; \uend\}],\\ \{\texttt{account} \mapsto o_{\sf{a}}\})\end{array}
    \end{array}
\right\}
\end{align*}

In \cref{ex:semantics_bank} we identified the heap as:

\begin{align*}
h = \left\{\begin{array}{l}
    o_{\sf{main}} \mapsto \begin{array}[t]{l}(\sf{Main},  \{\texttt{account} \mapsto o_{\sf{acc}},\\ \texttt{manager} \mapsto o_{\sf{man}}, \texttt{db} \mapsto o_{\sf{d}}, \})\end{array} \\
    o_{\sf{acc}} \mapsto (\sf{BankAccount}, \{\texttt{amount} \mapsto 0\}) \\
    o_{\sf{man}} \mapsto (\sf{SalaryManager}, \{\texttt{account} \mapsto o_{\sf{a}}\})\\
    o_{\sf{db}} \mapsto (\sf{DataStorage}, \{\texttt{account} \mapsto o_{\sf{a}}\})
    \end{array}
\right\}
\end{align*}

We have $\Gamma \vdash e : \sf{void} \dashv \Gamma''$ where $\Gamma''$ is the terminated environment containing the objects of $\Gamma$. Using the (ctx) and (call-ind) rule we can conclude the following transition (we let $e'$ denote the updated expression):

\begin{align*}
\<h, o_{\sf{main}}.\texttt{manager}.\texttt{add}&\texttt{Salary}(100.0);o_{\sf{main}}.\texttt{db}.\texttt{store}(\texttt{unit})\> \\ 
\ss[o_{\sf{man}}.\texttt{addSalary}] 
\<h, (&o_{\sf{man}}.\texttt{account}.\texttt{setMoney}(100.0);\\&o_{\sf{man}}.\texttt{account}.\texttt{applyInterest}(1.05));\\&o_{\sf{main}}.\texttt{db}.\texttt{store}(\texttt{unit})\>
\end{align*}

Now let $\Gamma'$ be the updated environment where a single transition has been performed on the salary manager object: 

\begin{align*}
\Gamma' = \left\{\begin{array}{l}
    o_{\sf{main}} \mapsto \begin{array}[t]{l}(\sf{Main}[\uend],\\ \begin{array}[t]{l}\{\texttt{account} \mapsto o_{\sf{acc}}, \texttt{manager} \mapsto o_{\sf{man}},\\ \texttt{db} \mapsto o_{\sf{d}}, \})\end{array}\end{array} \\
    o_{\sf{acc}} \mapsto \begin{array}[t]{l}(\sf{BankAccount}[\begin{array}[t]{l}\{\sf{setMoney};\\ \{\sf{applyInterest}; \\\{ \sf{getMoney}; \uend\}\}\}],\end{array}\\ \{\texttt{amount} \mapsto \texttt{double}\})\end{array} \\
    o_{\sf{man}} \mapsto \begin{array}[t]{l}(\sf{SalaryManager}[\uend],\\ \{\texttt{account} \mapsto o_{\sf{a}}\})\end{array}\\
    o_{\sf{db}} \mapsto \begin{array}[t]{l}(\sf{DataStorage}[\{\sf{store}; \uend\}],\\ \{\texttt{account} \mapsto o_{\sf{a}}\})\end{array}
    \end{array}
\right\}
\end{align*}

We can conclude $\Gamma \trans[o_{\sf{man}}.\texttt{addSalary}] \Gamma'$ with the (trans) rule. It is clear that we have $\Gamma' \vdash h$ since we have only updated a usage which is not considered in the consistency relation. Finally we can also conclude $\Gamma' \vdash e' : \sf{void} \dashv \Gamma''$ since the remaining usages in $\Gamma'$ corresponds to the remaining method calls in $e'$ (and also directly from the typing rule of $\Gamma \vdash e : \sf{void} \dashv \Gamma''$).

\end{example}

As previously mentioned, we use the labels of the run-time semantics to establish a correspondence between updates to the typing environment and the run-time configurations. In \cref{cor:usage_conformance}, which follows from \cref{thm:subject-reduction}, we make this correspondence explicit by showing that when a method call or label selection occurs at run-time, this always follows the protocol of the object. 

\begin{corollary}[Protocol conformance]
\label{cor:usage_conformance}
If $\Gamma \vdash e : T \dashv \Gamma''$, $\Gamma \vdash h$, $\<h, e\> \ss[o.\alpha] \<h', e'\>$ then $\exists \Gamma'.\ \Gamma' \vdash e' : T \dashv \Gamma''$ and $\Gamma(o).\sf{usage} \trans[\alpha] \Gamma'(o).\sf{usage}$
\end{corollary}

\begin{lemma}[Protocol completion]
\label{lemma:protocol_completion}
Let $\overline{D}$ be a well-typed program and let $c$ be the initial configuration of $\overline{D}$. If $c \ss^* \<h, v\>$ then all objects in $h$ has finished their protocol. 
\end{lemma}
\begin{proof}
Since $\vdash \overline{D}\ \texttt{ok}$ we know from (Main) that $\sf{Main} \{\U, \overline{F}, \overline{M}\} \allowbreak\in \overline{D}$, $\overline{M} = \{\sf{fun}\ \sf{main}()\ \{e\}\}$, and $\{o_{\sf{main}} \mapsto (\sf{Main}[\uend], \overline{F}.\sf{inittypes})\} \allowbreak\vdash e : T \dashv \Gamma'$ where $\sf{term}(\Gamma')$. Since $\sf{term}(\Gamma')$ we know that all objects have terminated protocols, and from \cref{cor:usage_conformance} we know that all objects has followed their protocols.
\end{proof}

We can now conclude with progress property, which states that well-typed programs do not get stuck.

\begin{restatable}[Progress]{theorem}{progress}
\label{thm:progress}
If $\Theta;\Omega;\Gamma \vdash e:T \dashv \Gamma'$, $\Gamma \vdash h$, then either $e$ is a value or $\exists h', e'.\ \<h, e\> \ss \<h', e'\>$
\end{restatable}

\begin{proof}
See details in \cref{app:proof-progess}.
\end{proof}

\begingroup
\lstset{language=scalaprotocol}

\section{Implementation}
\label{sec:implementation}
In this section we present Papaya, an implementation of our framework and type system for a subset of the Scala programming language \cite{Odersky2021ScalaSpecification}. The  source code can be found on our Github repository \cite{ravier_2021}.

The Papaya tool is implemented as a plugin for the Scala compiler, meaning that programs with protocol violations will produce compilation errors and the program will not be compiled. 

Scala is an object-oriented language which compiles to JVM bytecode and consequently is compatible with Java code and applications, while introducing new language features such as lazy evaluation, immutability, type inference and pattern matching. %
These features make Scala an expressive and powerful high-level language that supports object-oriented programming, functional programming, and a mix of both. %

The main features supported by Papaya are:
\begin{itemize}
    \item \textbf{Control flow structures} Papaya supports users using loops, if-else statements, match statements, and functions.
    \item \textbf{Recursion} Papaya handles recursive function calls as described for the formalisation. 
    \item \textbf{Fields} Objects with typestates can be stored in class fields and dealt with appropriately.
    \item \textbf{Unrestricted aliasing} Papaya offers the user unrestricted aliasing of variables.
\end{itemize}

To compare with the earlier implementations of Mungo for Java, we can see that Papaya introduces new features:
\begin{itemize}
    \item \textbf{Uncertain states} For increased flexibility, Papaya allows multiple branches to result in different typestates, as long as the type state of all objects eventually become consistent. This is a deviation between the formalism of this paper and the implementation.
    \item \textbf{Unrestricted aliasing} Mungo enforces linearity in its program, disallowing the user to alias objects with a protocol. With unrestricted aliasing in Papaya, the user is free to alias as much as they want to.
\end{itemize}

The implemented algorithm follows the structure defined for the formalism, by analysing the program graph, starting from entry-point into the program and expanding the program graph upon reaching method calls. At each method call encountered during verification, the transition system of the protocol of the callee is consulted to ensure that the method call is currently allowed for the object.

This is different from the previous implementations of the Mungo tool. In the original implementation \cite{Kouzapas2016TypecheckingStMungo,Dardhaetal2017,Kouzapas2018TypecheckingJava,VoineaDG20} the tool infers the typestate of objects in the program and then checks that this respects the typestate defined for the classes. Since that version of the Mungo tool requires a linear treatment of references, the tool works similar to a classic type system for object-oriented language where each class is checked in isolation.
In a recent implementation of Mungo \cite{Mota2021}, the Java Checker Framework is used to analyse the control flow graph of a Java program, and perform typestate analysis. The tool allows for a modular approach to type checking, however the use of aliases in the new implementation is still restricted to fractional permissions where write-access is only given to linear references.

\begin{example}
We return to our bank account example. This time we write a working implementation in Scala and use Papaya to verify the correctness of the implementation. 

In \cref{lst:annotation} we show the implementation of the bank account introduced earlier. The typestate is specified using the \lstinline{@Typestate} annotation where the argument refers a name of a singleton object defining the behaviour of a class. 

\begin{lstlisting}[style=color,language=scalaprotocol,label={lst:annotation}, caption={Implementation of the BankAccount with attached protocol}]
@Typestate("BankAccountProtocol")
class BankAccount() {
    var balance:Float = 0
    def fill(amount:Float):Unit = 
        { balance = amount }
    def get():Float = balance
    def applyInterest(ir:Float):Unit = {
        balance = balance * ir
    }
}
\end{lstlisting}

The protocol is written in a Scala-like domain specific language. The protocol of the bank account is shown in \cref{lst:protocol}. Notice that the implementation uses state equations (i.e. init = setMoney(Float) → intermediate) instead of the recursive definitions used in the formalism to describe state changes. This change is introduced to allow programmers to specify their protocols more easily.

\begin{lstlisting}[style=color,language=scalaprotocol,label={lst:protocol}, caption={Protocol for the BankAccount class}]
object BankAccountProtocol extends ProtocolLang with App {
  in("init")
    when("setMoney(Float)") 
        goto "intermediate"
  in("intermediate")
    when("applyInterest(Float)") 
        goto "filled"
  in("filled")
    when("getMoney()") 
        goto "end"
  in("end")
    end()
}
\end{lstlisting}

The protocol specifies that the BankAccount starts in the \lstinline[breaklines=false]{"init"} state and can perform one transition with a call to \lstinline{setMoney(Float)} to go to the \lstinline{"intermediate"} state. We can see that it then has one possible transition to the \lstinline{"filled"} state, whence it has one last possible transition to the \lstinline{"end"} state. Comparing this to the previously defined usage $\{\sf{setMoney}; \{\sf{applyInterest}; \allowbreak \{\sf{getMoney};\allowbreak \uend\}\}\}$ we see that the two descriptions are equivalent.

\begin{lstlisting}[style=color,language=scalaprotocol,label={lst:remaining_classes}, caption={Implementations of two classes that will use a shared bank account},firstnumber=11]
class DataStorage() {
    var money:BankAccount = null;
    def setMoney(m:BankAccount):Unit = 
        { money = m}
    def store():Unit = {
        var amount = money.get()
        println(amount)
        // write to the database
    }
}
class SalaryManager() {
    var money:BankAccount = null;
    def setMoney(m:BankAccount):Unit = 
        { money = m}
    def addSalary(amount:Float):Unit = {
        money.fill(amount)
        money.applyInterest(1.02f)
    }
}
\end{lstlisting}

In \cref{lst:remaining_classes} we show the implementation of the two remaining classes previously introduced, and in \cref{lst:aliasing-scala} we show how aliasing is achieved by providing the \lstinline{account} reference to both the salary manager and the data store. 

\begin{lstlisting}[style=color,language=scalaprotocol,label={lst:aliasing-scala}, caption={Program segment that uses aliasing}]
object Demonstration extends App {
    val account = new BankAccount
    val manager = new SalaryManager
    val storage = new DataStorage
    manager.setMoney(account)
    storage.setMoney(account)
    manager.addSalary(5000)
    storage.store()
}
\end{lstlisting}

In the implementation we handle the layer of indirection between references (with potential aliasing) and objects similarly to the treatment in the type system in \cref{sec:typesystem} but with more information tracked in order to aid debugging and error handling. This means that the three references introduced in \cref{lst:aliasing-scala} are tracked independently but all point to the same underlying instance, as shown in \cref{fig:InstanceNAliases}.

\begin{figure}[htpb]
    \centering
    \includegraphics[width=.9\columnwidth,keepaspectratio, trim={2cm 18cm 7cm 2.5cm},clip]{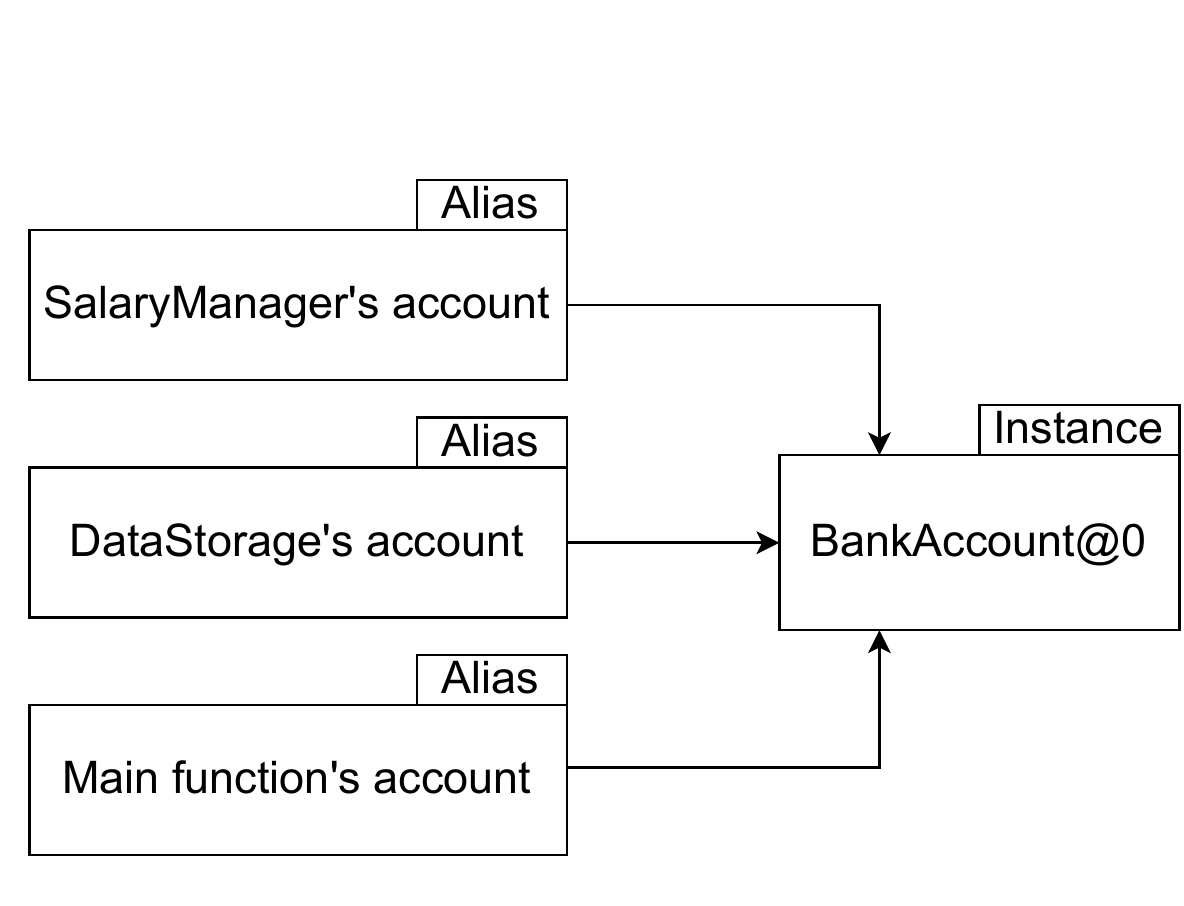}
    \caption[Example of one Instance with multiple Aliases pointing to it.]{Example of the structure of Instances and Aliases in the BankAccount example. Here we have three Aliases pointing to a single BankAccount Instance. The Instance has an "@0" ID to differentiate it from other potential BankAccount Instances. Each Alias is identified by its name and scope.}
    \label{fig:InstanceNAliases}
\end{figure}

\end{example}
\endgroup

\section{Related Work}
\label{sec:relatedwork}
\paragraph{Session types}
Session types \cite{Honda1993,Honda1998LanguageProgramming,Honda2008MultipartyTypes} were introduced to ensure type-safe structured communication between multiple parties. The process sending data and the process receiving data must agree on the type of data being transmitted. The concept of session types has been also explored for object oriented languages \cite{Dezani-Ciancaglini2006SessionLanguages, Vasconcelos2011}. A particular application of session types for an object oriented language is that in Bica \cite{Gay2010ModularProgramming} where session types are used to type communication on channels, but also to perform type-safe interaction with objects themselves. In terms of message passing in object-oriented languages, we can view a session type as a description of the messages we can send to a particular object, corresponding to an ordering of method calls. In the simplest setting we can imagine object initialisation as instantiating a communication channel between the new object and the caller, and subsequent method calls as sending messages on this channel. Scalas et al. \cite{ScalasY16,Scalas2017AProgramming} integrate binary and multiparty session types in Scala and implement it as a library.

\paragraph{Typestates}
While originally introduced to track value initialisation \cite{Strom1986Typestate:Reliability}, typestates have been explored extensively for object-oriented programming. The approach described for Bica is one approach for typestates in object-oriented languages that inspired the line of research on Mungo \cite{mungo_website, Gay2015ModularObjects, Kouzapas2018TypecheckingJava, BravettiBehaviouralTypes2020, Mota2021JavaChecker, Mota2021}. In this line of work, typestates are based on session types and describe the permitted sequence of method calls, in a syntax similar to session types.

Plaid \cite{Aldrich2009Typestate-orientedProgramming, Sunshine2011} introduces the concept of \textit{Typestate-oriented programming} wherein typestates form the basis for objects, rather than class descriptions. As operations are performed on object, they transition between states, and the set of available operations evolve, ensuring that methods can only be called on objects that are in a state that implements the method.

The Fugue protocol checker \cite{DeLine2004, Deline2004Objects} extends class definitions for the Common Language Runtime (CLR) \cite{microsoft_2020}  with state machines. They use pre and postconditions to describe the transitions between states and preconditions are used as guards, to ensure that methods are only called when the object is in a state that allows the method calls. 

Lastly, the work on typestates for concurrent object-oriented languages \cite{Padovani2018Deadlock-FreeProgramming, Crafa2017} uses typestates to reason about protocol conformance, but also properties such as deadlock freedom.

\paragraph{Aliasing and Typestates}
We have seen multiple approaches to combining typestates with object-oriented programming, but each approach handles the presence of aliasing differently.

Vault \cite{DeLine2001EnforcingSoftware} introduces \textit{tracked types} where a unique key is created for each object, and operations can only be performed on the object by the current holder of the key.

In an extension to the Vault language \cite{Fahndrich2002} the concepts of \textit{adoption} and \textit{focus} are used for introducing aliases. The adoption construct allows a linear value (the adoptee) to be converted into a nonlinear reference for the duration of the adopters lifetime. As linear resources of the adoptee cannot be accessed through the nonlinear type, they introduce the focus operation to temporarily convert the nonlinear type into a linear type, by ensuring that in the linear scope, no other aliases can witness the operations, and that the object is left in a consistent state after the operation, so that the operations remains invisible to other aliases. 

Later, in the work on Fugue \cite{DeLine2004} they allow objects to be marked \textit{NotAliased} and \textit{MayBeAliased}. In the case of an object being marked \textit{NotAliased} the object is treated linearly,  whereas objects marked \textit{NotAliased} are tracked to see if they can \textit{escape} from their context (by method calls or assignment, etc.) and emits a warning in case of unsafe aliasing.

Multiple approaches to aliasing have been introduced for the Plaid language. Bierhoff and Aldrich \cite{Bierhoff2007} present a fine-grain approach to aliasing. The authors note that an approach such as the one used in Fugue must be able to reason about all aliases to allow state change to an object, hence limiting nonlinear objects to simple operations. Instead they propose a collection of five permissions such as \textsf{unique} (single reference with read/write permissions), \textsf{share} (one reference has read/write permissions, other references has read permissions) or the inverse \textsf{pure} (read access while other reference has read/write permissions). For the different permissions, they introduce the concept of \textit{permission splitting} and \textit{permission joining}, where one alias with a permission can be split into two aliases that are equally or more restricted than the original. Similarly, for joining, two permissions can be merged back into a potentially less restrictive permission. To handle an arbitrary number of aliases, and ensure that all aliases can be collected to regain write access, they introduce fractions denoting how many times a permission has been split, and conversely when all fractions has been recovered.

A typesystem for a language inspired by Plaid \cite{Militao2010} uses concepts from behavioural separation \cite{Caires2013} to reason about type-states. In this language, classes are composed of \textit{views}, and each view contains a subset of the fields of the class. Through \textit{view equations}, views can be composed or decomposed into a number of other views, similar to permission splitting and joining as previously described. Through view decomposition, each alias is associated with a single view, and hence also follows the view equations. Similar to the previous work on aliasing in Plaid they use fractions to keep track of splits when allowing an unbounded number of aliases, so they can ensure that all aliases are recovered before any updates to the full object.  

Mungo generally treats objects as linear values, where only a single reference to an object can exist. While enforcing linearity allows for a common treatment of all object references, it is a deviation from real-world programs where aliasing is used in programming patterns for sharing data etc. Accordingly, work has been undergoing to lessen this constraint. A recent implementation of the Mungo tool \cite{Mota2021JavaChecker} supports access permissions similar to those described for Plaid. 

In another treatment of aliasing for Mungo \cite{Jakobsen2020, Golovanov2021}, the language of usages is extended with a parallel construct $(\U_1 \mid \U_2).\U_3$ where an object can be aliased into two references, with usages $\U_1$ and $\U_2$ respectively. After completion of the local protocol, only a single reference (with usage $\U_3$) exists. This approach is analogous to the view-equations used in \cite{Militao2010}. 

Common between the approaches to aliasing described in this section is that they adopt a local treatment of aliasing, allowing them to preserve compositionality of the type system, whereas the treatment in this paper is global. The local treatment allows for greater flexibility in a larger system, where components can be replaced without having to re-verify the entire system, whereas the global approach allows for the maximum flexibility for the programmer's work with aliasing.

A typestate verification framework for Java with support for aliasing has been presented in \cite{Fink2008EffectiveAliasing}. The tool makes sound approximations to scale to larger programs, at the cost of precision (increased false positives).

\section{Conclusion and Future Work}
\label{sec:conclusion}

In this paper we have explored a global approach to reasoning about unrestricted aliasing in the presence of typestates. We have shown the standard soundness properties about the type system, namely subject reduction and progress. Furthermore, we have shown the protocol conformance property--which ensures that protocols defined for classes are respected by instantiated objects, and that no protocol deviation occurs--and the protocol completion property--which ensures that protocols are completed for all objects, meaning that after termination of a program all objects have successfully completed their protocol.

The language presented in this paper is a small object-oriented language that does not correspond directly to any real-life programming language. However it does have similarities to the low level JVM bytecode language. As future work, we plan to explore this similarity in an attempt at integrating typestates in JVM bytecode. 

As we use a global approach of type checking the entire program graph, as opposed to checking each class in isolation, the run-time may suffer for larger programs. To combat this, it would be interesting to split classes into a linear section, and an unrestricted section. Then values that are treated by the class as linear objects (where only a single reference exists at all times) can be checked in isolation, before the global analysis checks the unrestricted sections of all classes. We leave it to future work to check whether such a split of a class can be determined without programmer annotations, and to explore how to integrate the previous approaches to type checking linear objects can be integrated as a step before the global analysis.

\section{Acknowledgements}
Research supported by the EPSRC programme grant ``From Data Types to Session Types: A Basis for Concurrency and Distribution" EP/K034413/1 (ABCD), and EU HORIZON 2020 MSCA RISE project 778233 ``Behavioural Application Program Interfaces'' (BehAPI).
We thank Simon Fowler for his valuable comments on the paper, Alceste Scalas for his helpful tips on Scala and Elena Giachino for her (implicit) suggestion on the name Papaya. 
\bibliography{references}

\appendix 

\onecolumn

\section{Proof for Unfolding Labelled Expressions}

\begin{lemma}[Weakening of label-environment]
\label{lemma:weakening}
If $\Theta;\Omega;\Gamma \vdash e : T \dashv \Gamma'$ and $k\not\in FL(e)$ then $\Theta;\Omega, k \mapsto \Gamma'';\Gamma \vdash e : T \dashv \Gamma'$
\end{lemma}
\begin{proof}
Simple structural induction in $e$.
\end{proof}

\begin{lemma}[Strengthening of label-environment]
\label{lemma:strengthening}
If $\Theta;\Omega, k \mapsto \Gamma'';\Gamma \vdash e : T \dashv \Gamma'$ and $k\not\in FL(e)$ then $\Theta;\Omega;\Gamma \vdash e : T \dashv \Gamma'$
\end{lemma}
\begin{proof}
Simple structural induction in $e$.
\end{proof}

\begin{lemma}[Substitution]
\label{lemma:substitution}
\noindent If 
\begin{itemize}
    \item $\Theta;\Omega;\Gamma \vdash k : e : \sf{void} \dashv \Gamma'$
    \item $k : e$ is well formed, 
    \item $\Theta';\Omega', k \mapsto \Gamma; \Gamma'' \vdash e' : T \dashv \Gamma'''$,
    \item $k \not\in BL(e')$, 
    \item $e' \in SUB(e)$, and
    \item $\Omega \subseteq \Omega'$
\end{itemize} 
Then 
\begin{itemize}
    \item $\Theta';\Omega';\Gamma'' \vdash e'\substitute{\sf{continue}\ k}{k : e} : T \dashv \Gamma'''$
\end{itemize}
\end{lemma}
\begin{proof}

\case{new, unit, field, par, null, true, false, enum}

No substitution occurs, hence it follows from \cref{lemma:strengthening}.

\bigskip 

\case{Assign}
Assume $\Theta';\Omega', k \mapsto \Gamma;\Gamma'' \vdash o.f = e'' : T \vdash \Gamma'''$. We know from (Assign) $\Theta';\Omega', k \mapsto \Gamma; \Gamma'' \vdash e'' : T' \dashv \Gamma''''$ and $\Gamma''' = \Gamma''''[o.f \mapsto \sf{vtype}(T)]$. From the induction hypothesis we know $\Theta';\Omega', e''\substitute{\sf{continue}\ k}{k : e} : T' \dashv \Gamma''''$, hence we can use (Assign) to conclude $\Theta';\Omega';\Gamma'' \vdash o.f = e'' : \sf{void} \dashv \Gamma'''$.

\bigskip 

\case{Call-d}

Assume $\Theta';\Omega', k \mapsto \Gamma;\Gamma'' \vdash o.m(e'') : T \vdash \Gamma'''$. From (Call-d) we know that $\Theta';\Omega', k \mapsto \Gamma;\Gamma'' \vdash e'' : T' \vdash \Gamma''''$. With the induction hypothesis we can conclude $\Theta';\Omega';\Gamma'' \vdash e''\substitute{\sf{continue}\ k}{k : e} : T' \vdash \Gamma''''$, hence we can conclude with (Call-d) that $\Theta';\Omega';\Gamma'' \vdash o.m(e'')\substitute{\sf{continue}\ k}{k : e} : T \vdash \Gamma'''$.

\bigskip 

\case{Call-ind}
Case similar to previous.

\bigskip 

\case{If}
Assume $\Theta';\Omega', k \mapsto \Gamma;\Gamma'' \vdash \sif{e''}{e'''}{e''''} : T \vdash \Gamma'''$. From (If) we know $\Theta';\Omega', k \mapsto \Gamma;\Gamma'' \vdash e'' : \sf{Bool} \vdash \Gamma''''$, $\Theta';\Omega', k \mapsto \Gamma;\Gamma'''' \vdash e''' : \sf{Bool} \vdash \Gamma'''$, and $\Theta';\Omega', k \mapsto \Gamma;\Gamma'''' \vdash e'''' : \sf{Bool} \vdash \Gamma'''$. Using the induction hypothesis three times, we get $\Theta';\Omega'\Gamma'' \vdash e''\substitute{\sf{continue}\ k}{k : e} : \sf{Bool} \vdash \Gamma''''$, $\Theta';\Omega';\Gamma'''' \vdash e'''\substitute{\sf{continue}\ k}{k : e} : \sf{Bool} \vdash \Gamma'''$, and $\Theta';\Omega';\Gamma'''' \vdash e''''\substitute{\sf{continue}\ k}{k : e} : \sf{Bool} \vdash \Gamma'''$, allowing us to use (If) to conclude $\Theta';\Omega';\Gamma'' \vdash \sif{e''}{e'''}{e''''}\substitute{\sf{continue}\ k}{k : e} : T \vdash \Gamma'''$.
\bigskip 

\case{Match}
Similar to previous case.
\bigskip 

\case{Label}
Assume $\Theta';\Omega', k \mapsto \Gamma; \Gamma'' \vdash k' : e'' : T \dashv \Gamma'''$. Since $k \not\in BL(k' : e'')$ we know that $k \neq k'$. From (Lab) we know $\Theta';\Omega', k \mapsto \Gamma, k' \mapsto \Gamma'';\Gamma'' \vdash e'' : \sf{void} \dashv \Gamma'''$. From our induction hypothesis we can conclude $\Theta';\Omega', k' \mapsto \Gamma'';\Gamma''\vdash e''\substitute{\sf{continue}\ k}{k : e} : \sf{void} \dashv \Gamma'''$, hence we can use (Lab) to conclude $\Theta';\Omega'; \Gamma'' \vdash k' : e''\substitute{\sf{continue}\ k}{k : e} : T \dashv \Gamma'''$.

\bigskip 

\case{Seq}

Assume $e' = e'';e'''$. Since $\Theta';\Omega', k \mapsto \Gamma;\Gamma'' \vdash e'';e''' : T \dashv \Gamma''' $ we know from (Seq) that $\Theta';\Omega', k \mapsto \Gamma;\Gamma'' \vdash e'' : T' \dashv \Gamma''''$ and $\Theta';\Omega', k \mapsto \Gamma;\Gamma'''' \vdash e''' : T \dashv \Gamma''$. We can use the induction hypothesis to conclude $\Theta';\Omega';\Gamma'' \vdash e''\substitute{\sf{continue}\ k}{k : e} : T' \dashv \Gamma''''$ and $\Theta';\Omega';\Gamma'''' \vdash e'''\substitute{\sf{continue}\ k}{k : e} : T \dashv \Gamma''$, hence it follows from (Seq) that 
$\Theta';\Omega';\Gamma'' \vdash e'';e'''\substitute{\sf{continue}\ k}{k : e} : T \dashv \Gamma''' $.

\bigskip 

\case{Continue}

If $e' = \sf{continue}\ k'$ where $k\neq k'$ then no substitution occurs, and the lemma is trivially true. So now assume $e' = \sf{continue}\ k$, hence we must show that $\Theta';\Omega';\Gamma'' \vdash k : e : \sf{void} \dashv \Gamma'''$.
From (con) we know that $\Gamma'' = \Gamma$. Since $e$ is well-formed, and $e' \in SUB(e)$ then the free choice of $T$ and $\Gamma'''$ is restricted to $\sf{void}$ and $\Gamma'$ respectively. The reason for this is, that well-formedness ensures that any continue-expression is guarded by a branching statement with a least one non-terminating branch. And as all branches must result in the same type and final environment, the type and final environment from the continue expression must be chosen to match the environment and type of the terminating branch. Hence it remains to show $\Theta';\Omega';\Gamma \vdash k : e : \sf{void} \dashv \Gamma'$. This follows from \cref{lemma:weakening}, since we know that $\Omega' \subseteq \Omega$ and $k \not\in FL(e)$ due to well-formedness.
\end{proof}

\begin{lemma}[Unfolding]
\label{lemma:unfolding}
Assume $\Theta;\Omega;\Gamma \vdash k : e : \sf{void} \dashv \Gamma''$, $\Gamma \vdash h$, and $\<h, k : e\> \ss \<h, e\substitute{\sf{continue}\ k}{k :e}\>$. Show that $\exists \Gamma'$ such that $\Theta;\Omega;\Gamma' \vdash e\substitute{\sf{continue}\ k}{k :e} \dashv \Gamma''$. 
\end{lemma}
\begin{proof}
From (Lab) we know that $\Theta;\Omega, k \mapsto  \Gamma;\Gamma \vdash e : \sf{void} \dashv \Gamma''$. Because $k : e$ is well-formed, we know that $k \not\in BL(e)$. We can then use \cref{lemma:substitution} to conclude that $\Theta;\Omega;\Gamma \vdash e\substitute{\sf{continue}\ k}{k :e} \dashv \Gamma''$
\end{proof}
\section{Proof of Unfolding Recursive Calls}
\label{app:recursion_unfolding_proof}

\recursionlemma*

\begin{proof}
Induction in $e$.

\case{Call} Assume $e = o'.m'(e')$. If $m \neq m'$ or $o \neq o'$ then it is not a recursive call. From (Call-d) we know $\Theta, o.m \mapsto \typingbound; \emptyset; \Gamma \vdash e' : T \dashv \Gamma''$, $\Gamma''(o')=(o'[C, \U], \lambda)$, $\U \trans[m'] \U'$, $\sf{agree}(t, T)$, $\sf{fun}\ m(x : t) : t' \{\exprbody \} \in D(C).\sf{methods}$, and $\Theta, o.m \mapsto \typingbound, o'.m' \mapsto \Gamma'''; \emptyset; \Gamma''' \vdash \exprbody \substitute{\sf{this}}{o'}\substitute{x}{\sf{getValue(t)}} : T' \dashv \typingfinal$, where $\Gamma''' = \Gamma''[o' \mapsto (o'[C, \U'], \lambda)]$.

From the induction hypothesis we get $\Theta; \emptyset; \Gamma \vdash e' : T \dashv \Gamma''$ as well as $\Theta, o.m \mapsto \typingbound; \emptyset; \Gamma''' \vdash \exprbody \substitute{\sf{this}}{o'}\substitute{x}{\sf{getValue(t)}} : T' \dashv \typingfinal$, and by applying it one more time we get $\Theta ;\emptyset; \Gamma''' \vdash \exprbody \substitute{\sf{this}}{o'}\substitute{x}{\sf{getValue(t)}} : T' \dashv \typingfinal$.

Otherwise if $o' = o$ and $m = m'$ then from (Call-d-rec) we have $\Theta, o.m \mapsto \typingbound; \emptyset; \Gamma \vdash e' : T \dashv \Gamma''$, $\Gamma''(o) = (o[C, \U], \lambda)$, $\U \trans[m] \U'$, $\sf{agree}(t, T)$, $\sf{fun}\ m(x : t) : t' \{\exprbody \} \in D(C).\sf{methods}$ and $\typingbound = \Gamma''[o \mapsto (o[C, \U'], \lambda)]$. From the IH we get $\Theta; \emptyset; \Gamma \vdash e' : T \dashv \Gamma''$.

By inversion we must have added the binding to $\Theta$ in a (Call-d), hence we must have had $\Theta' \emptyset; \Gamma^{(4)} \vdash \exprorig : \typeorig \dashv \typingorig$ where $\Theta' \subseteq \Theta$. From (Call-d) we would then have $\Theta' \emptyset; \Gamma^{(4)} \vdash \exprparorig : T''' \dashv \Gamma^{(5)}$, and $\typingbound = \Gamma^{(5)}[o \mapsto (o[C, \U'], \lambda)$. Due to well-formedness, we know that no expresions can follow a recursive call and that all recursive calls are guarded, hence the resulting typing environment and type of a recursive call must be chosen such that it matches the terminating branch of the body, hence we must have that $T = \typeorig$ and $\typingorig$. So with weakening of $\Theta'$ we can conclude $\Theta; \emptyset; \Gamma \vdash o.m(e') : T \dashv \typingfinal$.

The case for indirect calling is similar

\case{Remaining cases} all remaining cases are trivial or follows directly from the induction hypothesis.
\end{proof}

\section{Proof of Subject Reduction}
\label{app:proof-subject-reduction}

\subjectreduction*

\begin{proof}
Structural induction in $e$.

\case{Comp} Assume $\Gamma \vdash h$, $\emptyset;\emptyset;\Gamma \vdash e;e' : T \dashv \Gamma'$, and $\<h, e;e'\> \ss \<h', e''\>$.

If $\<h, e;e'\> \ss[\varepsilon] \<h', e''\>$ was concluded using rule (Seq) then $e'' = e'$, $e=v$, and $h'=h$. From the rules (Unit), (Enum), (Object), (Null), and (Bool) we see that $\emptyset;\emptyset;\Gamma \vdash v : T' \dashv \Gamma$. We know from our assumptions that $\emptyset;\emptyset;\Gamma \vdash e;e' : T \dashv \Gamma'$, hence from (Comp) we get that $\emptyset;\emptyset; \Gamma \vdash e' : T \dashv \Gamma'$. The last condition is trivial since $\Gamma \trans[\varepsilon] \Gamma$.

If, on the other hand,  $\<h, e;e'\> \ss[\alpha] \<h', e''\>$ was concluded using rule (ctx), then from the outermost evaluation context must have been $e;e' = \ctx[e'''];e'$ and we have $\<h, \ctx[e''']\> \ss[\alpha] \<h', \ctx[e'''']\>$. From our assumption $\emptyset;\emptyset;\Gamma \vdash e;e' : T \dashv \Gamma'$ we get from (Comp) that $\emptyset;\emptyset; \Gamma \vdash \ctx[e'''] : T' \dashv \Gamma''$. We can then, using the induction hypothesis, conclude that $\exists \Gamma'''.\emptyset;\emptyset;\Gamma'''\vdash \ctx[e''''] : T' \dashv \Gamma''$ such that $\Gamma \trans[\alpha] \Gamma'''$. By (Comp) we can finally conclude $\emptyset;\emptyset; \Gamma''' \vdash \ctx[e''''];e' : T \dashv \Gamma'$.

\bigskip

\case{New} Assume $\Gamma \vdash h$, $\emptyset;\emptyset;\Gamma \vdash o.f = \sf{new}\ C : T \dashv \Gamma'$, and $\<h, o.f = \sf{new}\ C\> \ss[\varepsilon] \<h', \sf{unit}\>$. From (New) we know that $T = \sf{void}$. From $\Gamma\vdash h$ we know that $o'$ is fresh for both $\Gamma$ and $h$, and hence it can be chosen as the reference for the new object, in both environments. 

From (New) we get that $\Gamma' = (\Gamma, o' \mapsto (o'[C, \U], C.\sf{fields}.\sf{inittypes}))[o.f \mapsto o']$. From (new) we know that $h'=(h, o' \mapsto (C, C.\sf{fields}.\sf{inittypes}))[o.f \mapsto o']$. Using (Unit) we can conclude $\emptyset;\emptyset;\Gamma' \vdash \sf{unit} : \sf{void} \dashv \Gamma'$. It is clear that $\Gamma \trans[\varepsilon] \Gamma'$ (case new). From the updates we have done to $\Gamma$ and $h$, it is also clear that $\Gamma'\vdash h'$.

\bigskip

\case{Field} Assume $\Gamma \vdash h$, $\emptyset;\emptyset;\Gamma \vdash o.f : T \dashv \Gamma'$, and $\<h, o.f\> \ss[\varepsilon] \<h', v\>$. From (Field) we know that $\Gamma(o).f = \sf{basetype}\ T$ (where $T \in \{\sf{bool}, \sf{void}, \bot\}$ or $\Gamma(o).f = \sf{reference}\ o$. 

We show the case for $v=\sf{true}$, the cases for \sf{false}, \sf{null}, and $l$ are similar. 

We know from $\Gamma \vdash h$ that $\Gamma(o).f = \sf{bool}$ and from (fld) that $h'=h$. Hence from (Field) we get that $T=\sf{bool}$ and $\Gamma'=\Gamma$. Finally from (Bool) we can conclude that $\emptyset;\emptyset;\Gamma \vdash \sf{true} \dashv \Gamma'$.

If $v = o'$, then from (fld) we know that $h(o).f = o'$ and from $\Gamma \vdash h$ that $\Gamma(o') = (o'[C, \U], \lambda)$, hence $T=o'[C, \U]$. But then we can use (Object) to conclude $\emptyset;\emptyset; \Gamma \vdash o' : T \dashv \Gamma'$.

In both cases, we can use (empty) to conclude $\Gamma \trans[\varepsilon] \Gamma'$.

\bigskip

\case{Call-D}
Assume $\Gamma \vdash h$, $\emptyset;\emptyset;\Gamma \vdash o.m(e) : T \dashv \Gamma'$, and $\<h, o.m(e)\> \ss[o.m] \<h', e'\>$. Assume that $e=v$, the other case will be proven afterwards. From (Call-d) we know $\Gamma(o)=(o[C, \U], \lambda)$, $\U \trans[m] \U'$, and $\emptyset, m \mapsto \Gamma;\emptyset;\Gamma[o \mapsto (o[C, \U'], \lambda)] \vdash e\substitute{this}{o}\substitute{x}{v} \dashv \Gamma'$. From (call-d) we know $h' = h$ and $e' = e\substitute{this}{o}\substitute{x}{v}$. From \cref{lemma:recursion_unfolding} we can conclude that $\emptyset;\emptyset;\Gamma[o \mapsto (o[C, \U'], \lambda)] \vdash e\substitute{this}{o}\substitute{x}{v} \dashv \Gamma'$. $\Gamma[o \mapsto (o[C, \U'], \lambda)] \vdash h$ clearly follows from $\Gamma \vdash h$ as only the usage of $o$ is updated in $\Gamma$. $\Gamma \trans[o.m] \Gamma[o \mapsto (o[C, \U'], \lambda)]$ follows from (case call).

Now assume $e \neq v$. Then $\<h, o.m(e)\> \ss[\alpha] \<h', e'\>$ must have been concluded with the (ctx) rule with the outermost evaluation context $o.m(e) = o.m(\ctx[e''])$, hence we have $\<h, \ctx[e'']\> \ss \<h', \ctx[e''']\>$. From (Call-d) we know that $\emptyset;\emptyset;\Gamma \vdash \ctx[e''] : T' \dashv \Gamma''$, hence by our induction hypothesis we know $\exists \Gamma'''.\emptyset;\emptyset;\Gamma''' \vdash \ctx[e'''] : T' \dashv \Gamma''$ such that $\Gamma \trans[\alpha] \Gamma'''$. By (Call-d) we can then conclude that $\emptyset;\emptyset;\Gamma'''\vdash o.m(\ctx[e''']) : T \vdash \Gamma'$.

\bigskip

\case{Assign}

Assume $\Gamma \vdash h$, $\emptyset;\emptyset; \Gamma \vdash o.f = e : \sf{void} \dashv \Gamma'$, and $\<h, o.f=e\>\ss \<h', e'\>$.

If $\<h, o.f=e\>\ss[\varepsilon] \<h', e'\>$ was concluded with (assign) then $e = v$, $e' =\sf{unit}$, and $h'=h[o.f\mapsto v]$. From (Assign) we know that $\emptyset;\emptyset;\Gamma \vdash v : T \dashv \Gamma''$, and $\neg\sf{only}(\Gamma'', o, f)$. From (Unit, Bool, Enum, Const, Obj) we know that $\Gamma'' = \Gamma$. Lastly we know that $\Gamma' = \Gamma[o.f \mapsto \sf{vtype}(T)]$. We see that the updates to $h'$ and $\Gamma'$ match exactly, hence we still have $\Gamma'\vdash h'$, and trivially from rule (Unit) we have $\emptyset;\emptyset; \Gamma' \vdash \sf{unit} : \sf{void} \dashv \Gamma'$. Finally we have that $\Gamma \trans[\varepsilon] \Gamma'$ from (case assign).

If on the other hand $\<h, o.f=e\>\ss \<h', e'\>$ was concluded with (ctx) then $o.f = e$ must be the evaluation context $(o.f = \ctx)[e'']$ and from (ctx) we have $\<h, \ctx[e'']\> \ss[\alpha] \<h', \ctx[e''']\>$. From (Assign) we know $\emptyset;\emptyset;\Gamma \vdash \ctx[e''] : T \dashv \Gamma''$, hence we can use the induction hypothesis to conclude $\exists \Gamma'''.\emptyset;\emptyset;\Gamma''' \vdash \ctx[e'''] : T \dashv \Gamma''$ such that $\Gamma \trans[\alpha] \Gamma'''$, hence using (Assign) we can conclude $\emptyset;\emptyset;\Gamma'''\vdash o.f=\ctx[e'''] : \sf{void} \dashv \Gamma'$.

\bigskip

\case{Case}

Assume $\Gamma \vdash h$, $\emptyset;\emptyset; \Gamma \vdash \sf{match}(e)\{\overline{l: e}\} : T \vdash \Gamma'$, and $\<h, \sf{match}(e)\{\overline{l: e}\}\>\ss \<h', e'\>$.

If $\<h, \sf{match}(e)\{\overline{l: e}\}\> \ss[o.l_j] \<h', e'\>$ was concluded with (match) then we know $e = o.l_j$ and $e'=e_j$ where $\overline{l: e} = l_1 : e_1, l_2 : e_2, \ldots l_i : e_i$ and $1 \leq j \leq i$. From (Case) we have that $\emptyset;\emptyset;\Gamma \vdash o.l_j : L\ \sf{link}\ o' \vdash \Gamma''$, and from (Enum) we know that $o'=o$ and $\Gamma''=\Gamma$. Furthermore, from (Case), we have that $\forall 1 \leq k \leq i. \Gamma(o).\sf{usage} \trans[l_k] \U_k \wedge \emptyset;\emptyset;\Gamma[o.\sf{usage} \mapsto \U_k] \vdash e_k : T \dashv \Gamma'$, hence we have that $\emptyset;\emptyset;\Gamma[o.\sf{usage} \mapsto U_j] \dashv e_j : T \dashv \Gamma'$. We have $\Gamma \trans[o.l_j] \Gamma[o.\sf{usage} \mapsto U_j]$ from (case label).

If, on the other hand, $\<h, \sf{match}(e)\{\overline{l: e}\}\> \ss[\alpha] \<h', e'\>$ was concluded with (ctx) then $\sf{match}(e)\{\overline{l: e}\}=\sf{match}(\ctx)\{\overline{l: e}\}[e'']$, and $\<h, \ctx[e'']\> \ss[\alpha] \<h', \ctx[e''']\>$. From (Case) we know $\emptyset;\emptyset;\Gamma \vdash \ctx[e''] : L\ \sf{link}\ o \dashv \Gamma''$, and since $\Gamma \vdash h$ we can conclude, with the induction hypothesis, that $\exists \Gamma'''.\emptyset;\emptyset;\Gamma'''\vdash \ctx[e'''] : L\ \sf{link}\ o \dashv \Gamma''$ such that $\Gamma \trans[\alpha] \Gamma'''$, hence by (Case) we can conclude $\emptyset;\emptyset;\Gamma''\vdash e' \dashv \Gamma'$.

\bigskip

\case{If}

Assume $\Gamma \vdash h$, $\emptyset;\emptyset;\Gamma \vdash \sif{e}{e_1}{e_2} : T \dashv \Gamma'$, and $\<h, \sif{e}{e_1}{e_2}\> \ss \<h', e'\>$.

If $\<h, \sif{e}{e_1}{e_2}\> \ss[\varepsilon] \<h', e'\>$ was concluded using (if-true) then we know that $e=\sf{true}$, $e' = e_1$, and $h' = h$. Then it follows directly from (If) that $\emptyset;\emptyset;\Gamma \vdash e_1 : T \dashv \Gamma'$. Using (empty) we can conclude $\Gamma \trans[\varepsilon] \Gamma$ The case for (if-false) is similar.

If $\<h, \sif{e}{e_1}{e_2}\> \ss[\alpha] \<h', e'\>$ was concluded using (ctx) then we must have $\sif{e}{e_1}{e_2} = \sif{\ctx}{e_1}{e_2}[e'']$ and $\<h, \ctx[e'']\> \ss \<h', \ctx[e''']\>$. From (If) we know $\emptyset;\emptyset;\Gamma \vdash \ctx[e''] : \sf{bool} \dashv \Gamma''$, and from the induction hypothesis we can conclude $\exists  \Gamma'''.\emptyset;\emptyset;\Gamma''' \vdash \ctx[e'''] : \sf{bool} \dashv \Gamma''$ such that $\Gamma \trans[\alpha] \Gamma'''$. Finally, using (If) we can conclude that $\emptyset;\emptyset;\Gamma''' \vdash \sif{\ctx}{e_1}{e_2}[e'''] : T \dashv \Gamma'$.

\bigskip

\case{Label}

Assume $\Gamma \vdash h$, $\emptyset;\emptyset;\Gamma \vdash k : e : \sf{void} \dashv \Gamma'$, and $\<h, k : e\> \ss \<h', e'\>$.

We know that $\<h, k : e\> \ss \<h', e'\>$ must have been concluded with (lab) hence we know that $e'=\substitute{\sf{continue}\ k}{k : e}$ and $h'=h$. The remainder of this case follows from \cref{lemma:unfolding}.

\end{proof}
\section{Proof of Progress}
\label{app:proof-progess}
\progress* 

\begin{proof}

Structural induction in $e$.

\case{Assignment}

Assume $\Theta;\Omega;\Gamma \vdash o.f = e : \sf{void} \dashv \Gamma'$ and $\Gamma \vdash h$. Show that $\<h, o.f = e\> \ss \<h', e'\>$.

If $e = v$ then we can trivially conclude with (assign) that $\<h, o.f=v\> \ss \<h[o.f\mapsto v], \sf{unit}\>$.

Otherwise if $e \neq v$ then from (Assign) we know $\Theta;\Omega;\Gamma \vdash e : T \dashv \Gamma''$, hence by the induction hypothesis we have $\<h, e\> \ss \<h'', e''\>$ and since $(o.f = e) = (o.f = [\_])[e]$ we can use (ctx) to conclude $\<h, o.f = e\> \ss \<h'', o.f = e''>$.

\bigskip 

\case{New}

Assume $\Theta;\Omega;\Gamma \vdash o.f = \sf{new}\ C : \sf{unit} \dashv \Gamma'$ and $\Gamma \vdash h$. With (new) we conclude directly that $\<h, o.f = \sf{new}\ C\> \ss \<h', \sf{unit}\>$ where $o'$ is fresh and $h' = (h, o' \mapsto (C, C.\sf{fields}.\sf{initvals}))[o.f\mapsto o']$.

\bigskip 

\case{Seq}

Assume $\Theta;\Omega;\Gamma \vdash e;e' : T \dashv \Gamma'$ and $\Gamma \vdash h$.

If $e = v$ then with (seq) we conclude $\<h, v;e'\> \ss \<h, e'\>$.

Otherwise; from (Comp) we know $\Theta;\Omega;\Gamma \vdash e : T' \dashv \Gamma''$, hence by the induction hypothesis we have $\<h, e\> \ss \<h', e''\>$. Since $e;e' = ([\_];e')[e]$ we can use (ctx) to conclude $\<h, e;e'\> \ss \<h', e'';e'\>$.

\bigskip 

\case{Call-d}

Assume $\Theta;\Omega;\Gamma \vdash o.m(e) : T \dashv \Gamma'$ and $\Gamma \vdash h$.

If $e = v$ then from (Call-d) we know $\Gamma(o).\sf{class} = C$ and $\sf{fun}\ m(x : t) : t' \{e'\} \in C.\sf{methods}$. From $\Gamma \vdash h$ we know that $h(o).\sf{class} = C$, hence we can conclude with (call-d) that $\<h, o.m(v)\> \ss \<h, e'\substitute{\sf{this}}{o}\substitute{x}{v}\>$.

Otherwise from (Call-d) we know $\Theta;\Omega;\Gamma \vdash e : T' \dashv \Gamma''$. With the induction hypothesis we can conclude $\<h, e\> \ss \<h', e''\>$. Finally, since $o.m(e) = (o.m([\_]))[e]$, we can conclude with (ctx) that $\<h, o.m(e)\> \ss \<h', o.m(e'')\>$.

\bigskip 

\case{If}

Assume $\Theta;\Omega;\Gamma \vdash \sif{e}{e'}{e''} : T \dashv \Gamma'$ and $\Gamma \vdash h$.

If $e = v$ then from (If) we know $\Theta;\Omega;\Gamma \vdash v : \sf{bool} \dashv \Gamma''$. This must have been concluded with (Bool), hence we know $v \in \{\sf{true}, \sf{false}\}$. If $v = \sf{true}$ then we can conclude, using (if-true), $\<h, \sif{\sf{true}}{e'}{e''}\>\ss \<h, e'\>$. If $v = \sf{false}$ then we can conclude, using (if-false), $\<h, \sif{\sf{false}}{e'}{e''}\>\ss \<h, e''\>$.

If, on the other hand, $e \neq v$, then from (If) we know $\Theta;\Omega;\Gamma \vdash e : \sf{bool} \dashv \Gamma''$. From the induction hypothesis we get $\<h, e\> \ss \<h', e'''\>$. Since $\sif{e}{e'}{e''} = (\sif{[\_]}{e'}{e''})[e]$ we can use (ctx) to conclude $\<h, \sif{e}{e'}{e''}\>\ss\<h', \sif{e'''}{e'}{e''}\>$.

\bigskip 

\case{Match}

Assume $\Theta;\Omega;\Gamma \vdash \sf{match}(e)\{\overline{l : e}\} : T \dashv \Gamma'$ and $\Gamma \vdash h$.

If $e = v$ then from (Match) we know $\Theta;\Omega;\Gamma \vdash v : L\ \sf{link}\ o$ which must have been concluded with (Enum) hence $v = o.l_j$. Furthermore we know $\overline{l : e} = l_1 : e_1, \ldots l_i: e_i$ and $1 \leq j \leq i$. Using (match) we can conclude $\<h, \sf{match}(o.l_j)\{\overline{l : e}\}\> \ss \<h, e_j\>$.

Otherwise from (Match) we know $\Theta;\Omega;\Gamma \vdash e : L\ \sf{link}\ o \dashv \Gamma'$. By the induction hypothesis we know $\<h, e\> \ss \<h', e'\>$. Since $\sf{match}(e)\{\overline{l : e}\} = (\sf{match}([\_])\{\overline{l : e}\})[e]$ we can use (ctx) to conclude $\<h, \sf{match}(e)\{\overline{l : e}\}\> \ss \<h', \sf{match}(e')\{\overline{l : e}\}\>$.

\bigskip 

\case{Label}

Assume $\Theta;\Omega;\Gamma \vdash k : e : T \dashv \Gamma'$ and $\Gamma \vdash h$. Here we directly conclude using (lab) that $\<h, k : e\>\ss\<h, e\substitute{\sf{continue}\ k}{k: e}\>$.

\end{proof}

\end{document}